\documentclass[times,11pt,journal,onecolumn]{IEEEtran}

\usepackage{cite}
\usepackage{fixltx2e}
\usepackage[cmex10]{amsmath}
\usepackage{array}
\usepackage{wasysym}
\usepackage{dsfont}
\usepackage[english]{babel}                         
\usepackage[T1]{fontenc}
\usepackage{mathtools}
\usepackage{amssymb} 
\usepackage{enumerate}
\usepackage{bbm}
\usepackage{epsfig,syntonly}
\usepackage{verbatim,times}
\usepackage{epstopdf}
\usepackage{graphicx}
\usepackage{accents}
\usepackage{latexsym,fancyhdr,bm}
\usepackage{url}
\usepackage{pstricks-add}
\usepackage{pst-coil}
\usepackage{pst-node}
\usepackage{subfigure}

\DeclareMathSymbol{\widehatsym}{\mathord}{largesymbols}{"62}
\newcommand\lowerwidehatsym{%
  \text{\smash{\raisebox{-1.3ex}{%
    $\widehatsym$}}}}
\newcommand\fixwidehat[1]{%
  \mathchoice
    {\accentset{\displaystyle\lowerwidehatsym}{#1}}
    {\accentset{\textstyle\lowerwidehatsym}{#1}}
    {\accentset{\scriptstyle\lowerwidehatsym}{#1}}
    {\accentset{\scriptscriptstyle\lowerwidehatsym}{#1}}
}

\DeclareMathSymbol{\widetildesym}{\mathord}{largesymbols}{"62}

\makeatletter
\let\l@ENGLISH\l@english
\makeatother

\newcommand{\fset}{{\mathcal F}}

\newcommand{\iset}{{\mathcal I}}

\newcommand{\hset}{{\mathcal H}}
\newcommand{\lset}{{\mathcal L}}

\newtheorem{theorem}{Theorem}
\newtheorem{proposition}{Proposition}

\newcommand{\btheo}{\begin{theorem}}
\newcommand{\etheo}{\end{theorem}}
\newcommand{\bproof}{\begin{proof}}
\newcommand{\eproof}{\end{proof}}
\newtheorem{definition}[theorem]{Definition}
\newcommand{\bdefi}{\begin{definition}}
\newcommand{\edefi}{\end{definition}}
\newtheorem{fact}[theorem]{Fact}
\newcommand{\bprop}{\begin{fact}}
\newcommand{\eprop}{\end{fact}}
\newtheorem{corollary}[theorem]{Corollary}
\newcommand{\bcor}{\begin{corollary}}
\newcommand{\ecor}{\end{corollary}}
\newtheorem{example}[theorem]{Example}
\newcommand{\bex}{\begin{example}}
\newcommand{\eex}{\end{example}}
\newtheorem{lemma}[theorem]{Lemma}
\newcommand{\blemma}{\begin{lemma}}
\newcommand{\elemma}{\end{lemma}}
\newtheorem{remark}[theorem]{Remark}
\newcommand{\bremark}{\begin{remark}}
\newcommand{\eremark}{\end{remark}}
\newtheorem{conj}[theorem]{Conjecture}
\newcommand{\bconj}{\begin{conj}}
\newcommand{\econj}{\end{conj}}

\def\0{{\tt 0}} 
\def\1{{\tt 1}} 
\def\?{{\tt *}} 
\renewcommand{\mid}{\,|\,}

\newcommand{\dens}[1]{\mathsf{#1}}
\newcommand{\Ldens}[1]{\dens{#1}}

\newcommand{\dr}{{\mathtt r}}
\newcommand{\dl}{{\mathtt l}}

\allowdisplaybreaks

\begin{document}

\title{How to Achieve the Capacity \\ of Asymmetric Channels}
\author{Marco~Mondelli, S.~Hamed~Hassani, and~R\"{u}diger~Urbanke%
\thanks{M. Mondelli is with the Department of Electrical Engineering, Stanford University, USA
(e-mail: mondelli@stanford.edu).

S. H. Hassani is with the Department of Electrical and Systems Engineering, University of Pennsylvania, USA
(e-mail: hassani@seas.upenn.edu).

R. Urbanke is with the School of Computer and Communication Sciences,
EPFL, CH-1015 Lausanne, Switzerland
(e-mail: ruediger.urbanke@epfl.ch).

The paper was presented in part at the 52nd Annual Allerton Conference on Communication, Control, and Computing, 2014.
}
}

\maketitle

\begin{abstract}
We survey coding techniques that enable reliable
transmission at rates that approach the capacity of an arbitrary
discrete memoryless channel. In particular, we take the point of
view of modern coding theory and discuss how recent advances in
coding for symmetric channels help provide more efficient solutions
for the \emph{asymmetric} case. We consider, in more detail, three
basic coding paradigms.

The first one is \emph{Gallager's} scheme that consists of concatenating
a linear code with a non-linear \emph{mapping} so that the input distribution
can be appropriately shaped. We explicitly show that both polar
codes and spatially coupled codes can be employed in this scenario.
Furthermore, we derive a scaling law between the gap to capacity,
the cardinality of the input and output alphabets, and the required
size of the mapper.

The second one is an \emph{integrated scheme}
in which the code is used \emph{both} for source coding, in order
to create codewords distributed according to the capacity-achieving
input distribution, \emph{and} for channel coding, in order to provide
error protection. Such a technique has been recently introduced by
Honda and Yamamoto in the context of polar codes, and we show how
to apply it also to the design of sparse graph codes. 

The third paradigm is based on an idea of B\"ocherer and Mathar, and separates the two tasks of source coding and channel coding by a \emph{chaining construction} that binds together
several codewords. We present conditions for the source code and the channel code, and we describe how to combine {\em any} source code with {\em any} channel code that fulfill those conditions, in order to provide capacity-achieving schemes for asymmetric channels. In particular, we show that polar codes, spatially coupled codes, and homophonic codes are suitable as basic building blocks of the proposed coding strategy.

Rather than focusing on the exact details of the schemes, the purpose of this tutorial is to present different coding techniques that can then be implemented with many variants. There is no absolute winner and, in order to understand the most suitable technique for a specific application scenario, we provide a detailed comparison that takes into account several performance metrics. 

\end{abstract}

\begin{IEEEkeywords}
Asymmetric channel, capacity-achieving coding scheme, chaining construction, Gallager's scheme, polar codes, spatially coupled codes.
\end{IEEEkeywords}

\section{Introduction}

We survey various coding techniques for achieving the capacity of an \emph{asymmetric} channel. There are at least three approaches for this problem, and each of them has multiple variants. Before describing these asymmetric coding techniques, let us quickly review how to solve the problem in the symmetric scenario. 

\subsection{Symmetric Channel Coding: A Review}

Let $W$ be a symmetric
binary-input, discrete memoryless channel (B-DMC) and denote by
$C(W)$ its capacity. Two schemes that are capable
of achieving the capacity for this class of channels are polar codes
\cite{Ari09} and spatially coupled codes \cite{KRU13}.

Polar codes are closely related to Reed-Muller (RM) codes, as,
in both cases, the rows of the generator matrix of a code of block
length $n =2^m$ are chosen from the rows of the matrix $G_n =
\left[\begin{array}{cc} 1 & 0 \\ 1 & 1 \\ \end{array}\right]^{\otimes
m}$, where $\otimes$ denotes the Kronecker product. The crucial
difference between polar and RM codes lies in the choice of the
rows. The RM rule consists in choosing the rows with the largest
Hamming weights, whereas the polar rule depends on the channel. To
explain how this last rule works, let us recall that polar codes
are decoded with a successive cancellation (SC) algorithm that
makes decisions on the bits one-by-one in a pre-chosen order. In
particular, in order to decode the $i$-th bit, the algorithm uses
the values of all the previous $i-1$ decisions. Therefore, the polar
rule chooses the rows of $G_n$ that correspond to the most reliable synthetic channels when decoded in this
successive manner. A rule that interpolates between the polar and the RM one is proposed in \cite{mondelli-polarRM}, in order to improve the finite-length performance. 

The complexity of encoding and decoding of polar codes scales
as $O(n \log n)$ \cite{Ari09}. The code construction can
be performed with complexity $\Theta(n)$ \cite{TV13con, RHTT}. Furthermore,
by taking advantage of the partial order between the synthetic
channels, it is possible to construct a polar code by computing
the reliability of a sublinear fraction of the synthetic channels \cite{MHU17subl}. If a rate $R<C(W)$ is fixed and $n$ goes to
$\infty$, then the block error probability $P_e$ scales as $O(2^{-\sqrt{n}})$
\cite{ArT09}. If the error probability $P_e$ is fixed and $n$ goes to $\infty$, then the gap to capacity $C(W)-R$ scales as $O(n^{-1/\mu})$.
The scaling exponent $\mu$ is lower bounded by $3.579$ and upper
bounded by $4.714$ \cite{HAU14, GB14, MHU15}. It is conjectured that
the lower bound on $\mu$ can be increased up to $3.627$, i.e., to the value for the binary erasure channel (BEC). For a unified view on the performance analysis and on the scaling of the relevant parameters of polar codes, see \cite{MHU15}.

Spatially coupled LDPC (SC-LDPC) codes are constructed by connecting
together $L$ standard LDPC codes in a chain-like fashion and by
terminating this chain so that the decoding task is
easier near the boundary. The decoding is done using the belief-propagation (BP) algorithm that has linear complexity in the block length per iteration.
Given a desired gap to capacity $C(W)-R$, we can choose an appropriate
degree distribution for the component codes, as well as an appropriate
connection pattern of the chain and length of the code, so that the
resulting ensemble enables reliable transmission over any B-DMC \cite{KRU13}. For a discussion on the scaling behavior of spatially
coupled codes the reader is referred to \cite{OU14}.  In Section~\ref{sc-tradeoff}, we show via heuristic arguments that for SC-LDPC codes, if the error probability $P_e$ is fixed, the gap to capacity $C(W)-R$ scales as $O(n^{-1/3})$, i.e., the scaling exponent $\mu$ is equal to $3$.

Polar coding schemes have been generalized to arbitrary symmetric DMCs.
Channel polarization for $q$-ary input alphabets is discussed
in \cite{STA09} and more general constructions based on arbitrary
kernels are described in \cite{MT10}. Furthermore, polar codes have been built exploiting various algebraic structures on the input
alphabet \cite{PB13, SaP13, NT16, Nas17a, Nas17b}. On the contrary, the capacity-achieving nature of SC-LDPC codes (and even threshold saturation) has been proved only for binary-input channels and it remains an open problem for the transmission over arbitrary symmetric DMCs. Results concerning the iterative decoding threshold on the BEC for non-binary SC-LDPC code ensembles are presented
in \cite{UKS11, PAC13}, and the corresponding threshold saturation
is proved in \cite{AA16}. The threshold analysis under windowed
decoding is provided in \cite{WKMFC14}.    

\subsection{Asymmetric Channel Coding: Existing and Novel Techniques}

The simplest example of \emph{asymmetric} DMC is the Z-channel,
which is schematically represented in Figure \ref{fig:Zch}: the
input symbol $0$ is left untouched by the channel, whereas the input
symbol $1$ is flipped with probability $\varepsilon$. The basic
problem that we face when transmitting over this channel is that
(proper) linear codes impose a uniform input distribution, whereas
the capacity-achieving input distribution for an asymmetric channel
is, in general, non-uniform. Indeed, for the case of the Z-channel,
the capacity-achieving distribution assigns to the symbol $1$ a
probability $\varepsilon^{\varepsilon/\bar{\varepsilon}}
\cdot(1+\bar{\varepsilon}\varepsilon^{\varepsilon/\bar{\varepsilon}})^{-1}$,
where we set $\bar{\varepsilon} = 1-\varepsilon$~\cite[Formula
(5.10)]{RiU08}. This mismatch in the input distribution bounds the
achievable transmission rate away from capacity (for linear codes).

It is worth pointing out that, at least for binary inputs, the
optimal distribution is not too far from the uniform one, in the
sense that the capacity-achieving input distribution always has a
marginal in the interval $(1/e, 1-1/e)$~\cite{MaR91}. In addition,
a fraction of at most $1-\frac12 e \ln(2) \approx 0.058$ of capacity
is lost if we use the uniform input distribution instead of the
optimal input distribution \cite{MaR91}. This result was later
strengthened in \cite{ShF04}, where the Z-channel is proved to be
extremal in this sense. As for channels with more than 2 inputs,
the upper bound $1-1/e$ to the range of the capacity-achieving
distribution still holds~\cite{Lia04}, but the lower bound $1/e$
is false. Given that the loss incurred by using a uniform input
distribution is relatively modest, why do we care about the problem
of achieving the full capacity of asymmetric channels? First of
all, it is an interesting theoretical problem. Second, over time,
all communication systems are increasingly optimized to take full
advantage of their capabilities, and even small gains become
significant.

\begin{figure}[t]
    \centering
\psset{arrowscale=1}
\psset{unit=0.5cm}
\psset{xunit=1,yunit=1}
\begin{pspicture}(0,0)(10,10)
\psline[linecolor=black,linewidth=0.7pt]{->}(2,2)(8,2)
\psline[linecolor=black,linewidth=0.7pt]{->}(2,2.1)(8,5.9)
\psline[linecolor=black,linewidth=0.7pt]{->}(2,6)(8,6)
\rput[c](2,1.5){\small{$1$}}
\rput[c](7.9,1.5){\small{$1$}}
\rput[c](2,6.5){\small{$0$}}
\rput[c](7.9,6.5){\small{$0$}}
\rput[c](5,1.4){\small{$\bar{\varepsilon}$}}
\rput[c](5,3.5){\small{$\varepsilon$}}
\end{pspicture}
\caption{Schematic representation of the $Z$-channel with parameter $\varepsilon$.}
\label{fig:Zch}
\end{figure}
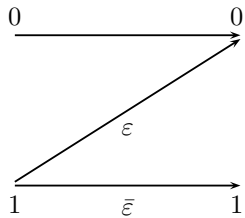

The classic solution to the problem of coding over asymmetric
channels goes back to Gallager and consists of concatenating a
linear code with a non-linear mapper so that the input distribution
becomes biased \cite{Gal68}. In \cite{McE01}, McEliece described
how this can be done successfully with iterative codes. We refer
to this approach as {\em Gallager's mapping} and we discuss
how any capacity-achieving coding scheme can be used for this
setting. In particular, by combining either polar codes or
spatially coupled codes with suitable non-linear mappers, we can
approach capacity arbitrarily closely. More specifically, we derive
a scaling law that relates the gap to capacity to the mismatch in
the actual input distribution and to the size of the channel
alphabets.

More recently, polar codes have been used to achieve the capacity
of binary-input asymmetric DMCs. In particular, in \cite{SRDR12}
the authors propose a solution that makes use of the concatenation
of two polar codes: one of them is used to solve a source coding
problem, in order to have codewords distributed according to the
capacity-achieving input distribution; the other is used to solve
a channel coding problem, in order to provide error correction.
However, such a scheme requires polarization for both the inner and
the outer codes, therefore the error probability scales roughly as
$2^{-n^{1/4}}$. Thus, in order to obtain the same performance as
standard polar codes, the square of their block length is required.
A very simple and more efficient solution for transmission over
asymmetric channels is presented in \cite{HY13} and this idea also
appears in \cite{GAG13ar}, where it is applied to the transmission
over the broadcast channel. More specifically, in order to transmit
over channels whose optimal input distribution is non-uniform, the
polar indices are partitioned into three groups: some are used for
information transmission; some are used to ensure that the input
distribution is properly biased; and some carry random bits shared
between the transmitter and the receiver.  The purpose of this
shared randomness is to facilitate the performance analysis.  Indeed,
as in the case of LDPC code ensembles, the error probability is
obtained by averaging over the randomness of the ensemble. In short,
the methods in \cite{SRDR12, HY13} exploit the fact that polar codes
are well suited not only for channel coding but also for lossless
source coding \cite{Ar10, CrKo10}. Clearly, this is not a prerogative
only of polar codes as, for example, sparse graph codes have been
successfully used for both channel coding and source coding purposes
\cite{CSV04dimacs}. Motivated by this fact, we describe a scheme
based on spatially coupled codes that achieves the capacity of
asymmetric DMCs by solving both a source coding and a channel coding
problem at the same time. A similar scheme was proposed in \cite{MuMi08} for lossy source coding and in \cite{MuMi09} for channel coding. Later, these ideas were unified and extended to several other problems (Slepian-Wolf, Gelfand-Pinsker, Wyner-Ziv, and One-helps-one problems), by using the notion of hash property of LDPC matrices \cite{MuMi10}. However, note that the existing works \cite{MuMi08,MuMi09,MuMi10} require MAP decoding, whereas we consider a low-complexity coding strategy. As it will be explained in detail later, this solution still does not have a formal proof, but it has
been verified in numerical simulations~\cite{AMV15, KVNP14}. We refer to this
technique as the {\em integrated scheme}, because we use
one code for both source {\em and} channel coding.

This brings us to the third coding paradigm. By ``chaining'' together
several codewords, we can decouple the problem of source coding
(creating a biased codeword from unbiased bits) from the problem
of channel coding (providing error correction).  The idea is based
on \cite{BoM11}, where the authors refer to it as the {\em bootstrap}
scheme. We prefer to use the name {\em chaining construction} that
was introduced in \cite{HRunipol}, where a similar approach was
used to design universal polar codes. The chaining construction is
a general method, and it was applied to devise polar coding schemes
for the wiretap channel \cite{SaV13} and for the broadcast channel
\cite{MHSU14}.  We show how to chain {\em any} suitable source
coding solution with {\em any} suitable channel coding solution,
in order to transmit over an asymmetric channel. We give explicit
conditions on the source and the channel code so that the overall
scheme is capacity-achieving. Furthermore, we prove that polar codes and homophonic codes can be used for the source coding part, and that polar codes and spatially coupled codes can be used for the channel coding part.

In a nutshell, this paper describes three different paradigms to
achieve the capacity of asymmetric channels and, as such, it is of
tutorial nature. Motivated by the recent advances in coding for
symmetric channels, we show that it is now possible to construct
efficient schemes also for the asymmetric case. As a result, we
demonstrate that perhaps what was once considered as a difficult
problem is in fact quite easy to solve with existing ``primitives''.
The three paradigms presented are quite general and should be
regarded as ``meta-schemes'' that can then be made more specific
by using a certain class of codes (e.g., polar codes or spatially coupled
codes) according to the particular scenario of interest. For this
reason, the interest of this paper is more in describing generic
coding ideas rather than in presenting formal proofs and providing
all the details for each scheme. With the objective of highlighting
the pros and cons of these coding approaches, we present them in a
unified manner and provide a detailed comparison  with a focus on
several features crucial in applications, such as, error probability,
rate penalty, computational complexity, universality, and
access to common randomness.

The rest of the paper is organized as follows. In Section~\ref{sec:twoprob}
we discuss two related and simpler problems, specifically, how to
achieve the {\em symmetric} capacity of an {\em asymmetric} channel
and how to perform error correction using biased codewords. The
proposed solutions to these problems will be later used as
``primitives'' to solve the transmission problem over asymmetric
channels. Then, we describe in three consecutive sections the coding
paradigms for achieving the capacity of an asymmetric DMC: Gallager's
mapping in Section~\ref{sec:gallager}, the integrated scheme in
Section~\ref{sec:integrated}, and the chaining construction in
Section~\ref{sec:chaining}. In Section \ref{sec:comp}, we provide
a comparison between these three different approaches. The metrics
taken into account are the error probability, the rate penalty, the
computational complexity, the universality, and the use
of common randomness. Finally, in Section \ref{sec:concl} we provide
some concluding remarks.

\section{Warm Up: Two Related Problems}\label{sec:twoprob}

After establishing the notation and reviewing some known concepts,
in this preliminary section we consider two related, but different,
problems that will be regarded as useful primitives in the following
sections. In particular, in Section \ref{subsec:symmcap}, we discuss how
to achieve the symmetric capacity of an asymmetric B-DMC, and in Section
\ref{subsec:checks} we describe how to transmit reliably
a biased binary codeword.

\subsection{Notation and Prerequisites}

Throughout this paper, we consider the transmission over a DMC $W: {\mathcal X}\to {\mathcal Y}$ with input alphabet $\mathcal X$ and output alphabet $\mathcal Y$. If the channel is binary-input, we usually take $\mathcal X = {\mathbb F}_2=\{0, 1\}$ and we say that $X$ is a Bernoulli$(\alpha)$ random variable if ${\mathbb P}(X=1) = \alpha$ for some $\alpha \in [0, 1]$. However, for the analysis of LDPC ensembles, it is convenient to consider the standard mapping $0 \longleftrightarrow 1$ and $1 \longleftrightarrow -1$. It will be clear from the context whether the input alphabet is $\{-1, 1\}$ or $\{0, 1\}$. The probability of the output being $y$ given an input $x$ is denoted by $W(y\mid x)$ and the probability of the input being $x$ given an output $y$ is denoted by $p_{X\mid Y}(x\mid y)$. We write $C(W)$ and $C_{\rm s}(W)$ to indicate the capacity and the symmetric capacity of $W$, respectively. Given the scalar components $X^{(i)}, \cdots, X^{(j)}$ and $X_{i}, \cdots, X_{j}$, we use $X^{i:j}$ as a shorthand
for the column vector $(X^{(i)}, \cdots, X^{(j)})^T$ and, similarly, $X_{i:j}$ as a shorthand for the column vector $(X_{i}, \cdots, X_{j})^T$ with $i \le j$. The index set $\{1, \cdots, n\}$
is abbreviated as $[n]$ and, given a set $\mathcal A\subseteq [n]$, we denote
by ${\mathcal A}^{\rm c}$ its complement. We denote by $\log$ and $\ln$ the logarithm in base $2$ and base $e$, respectively. For any $x\in [0, 1]$, we define $\bar{x} = 1-x$. The binary entropy function is given by $h_2(x) = -x\log x - \bar{x}\log\bar{x}$. When discussing sparse graph coding schemes, the parity-check matrix is denoted by $P$. We do not use $H$ to denote the parity-check matrix, as it is done more frequently, because the symbol $H(\cdot)$ indicates the entropy of a random variable. When discussing polar coding schemes, we assume that the block length $n$ is a
power of $2$, say $n=2^m$ for $m \in {\mathbb N}$, and we denote
by $G_n$ the polar matrix given by $G_n = \left[\begin{array}{cc}
1 & 0 \\ 1 & 1 \\ \end{array}\right]^{\otimes m}$, where $\otimes$ denotes
the Kronecker product.

Let us recall some notation and facts concerning B-DMCs. This part is telegraphic and the reader is referred to \cite{RiU08} for more details.
First, consider a symmetric
B-DMC with ${\mathcal X} = \{-1, 1\}$. Assume that $X$ is transmitted,
$Y$ is the received observation, and $L(Y)$ the corresponding
log-likelihood ratio, namely for any $y\in \mathcal Y$,
\begin{equation}\label{eq:loglike} L(y) = \ln \frac{W(y\mid 1)}{W(y\mid
-1)}.  \end{equation} Let us denote by $\Ldens{a}$ the density of
$L(Y)$ assuming that  $X=1$ and let us call it an $L$-density.

We say that an $L$-density $\Ldens{a}$ is symmetric if
\begin{equation}\label{eq:symm}
\Ldens{a}(y)= e^{y}\Ldens{a}(-y).
\end{equation}
Since the log-likelihood ratio constitutes a sufficient statistic
for decoding, two symmetric B-DMCs are equivalent if they have the
same $L$-density. A meaningful choice for the representative of
each equivalence class is $W(y\mid 1)=\Ldens{a}(y)$ and, by symmetry,
$W(y\mid -1)=\Ldens{a}(-y)$. Indeed, by using the assumption
\eqref{eq:symm}, we can show that this choice of $W(y\mid x)$
yields an $L$-density equal to $\Ldens{a}(y)$ (see \cite[Lemma
4.28]{RiU08}).

As a final reminder, the capacity $C(W)$ can be computed as a function
of the $L$-density $\Ldens{a}$ according to the following formula
\cite[Lemma 4.35]{RiU08}, \begin{equation}\label{eq:capacity} C(W)=
\int \Ldens{a}(y)\left(1-\log(1+e^{-y})\right)dy.  \end{equation}

\subsection{How to Achieve the Symmetric Capacity of Asymmetric Channels}\label{subsec:symmcap} 

\noindent {\bf Problem Statement.} Let $W$ be a (not necessarily symmetric) B-DMC. The aim is to transmit over $W$ with a rate close to $C_{\rm s}(W)$.

\vspace{1em}

\noindent {\bf Design of the Scheme.} The original construction of polar codes
directly achieves the symmetric capacity of any B-DMC \cite{Ari09}.

For sparse graph codes, some more analysis is required. 
Here, we will follow a line of reasoning inspired by \cite[Section 5.2]{RiU08}. A similar approach was first introduced in \cite{Huo03} and an alternative path that considers the average of the density
evolution analysis with respect to each codeword, is considered
in \cite{WKP03}. All these techniques lead to the same result.

The codebook of a code of block length $n$ and rate $R$ with parity
check matrix $P$ is given by the set of $x^{1:n} \in {\mathbb
F}_2^{n}$ s.t. $P x^{1:n} =0^{1:(1-R)n}$, where $0^{1:(1-R)n}$
denotes a column vector of $(1-R)n$ zeros. In words, the transmitter
and the receiver know that the results of the parity checks are all
zeros.  Let us consider a slightly different model in which the
values of the parity checks are chosen uniformly at random and this
randomness is shared between the transmitter and the receiver:
first, we pick the parity checks uniformly at random; then, we pick
a codeword uniformly at random among those that satisfy the parity
checks.  Clearly, this is equivalent to picking directly one codeword
chosen uniformly at random from the whole space ${\mathbb F}_2^{n}$.
As a result, we can model the codeword as a sequence of $n$ uniform
i.i.d. bits. Note that, in \cite{Huo03}, instead of randomizing the
cosets, the authors add a random scrambling vector to the entire
codeword before transmission and then subtract it afterwards. The
random scrambling and de-scrambling is absorbed into a normalized
channel that is automatically symmetric, hence the standard density
evolution equations hold. The concentration theorem (for a random
code, scrambling vector, and channel realization) also follows by
absorbing the scrambling bit into the randomness of the channel.
This scrambling idea is also used in \cite{W15al}, where the authors
explore the connection between symmetric channel coding, general
channel coding, symmetric Slepian-Wolf coding, and general Slepian-Wolf
coding.

Since the channel can be asymmetric, we need to define two distinct
$L$-densities according to the transmitted value. For simplicity,
let us map the input alphabet ${\mathbb F}_2$ into $\{-1, 1\}$ and
denote by $\Ldens{a}^+(y)$ and $\Ldens{a}^-(y)$ the $L$-density for
the channel assuming that $X=1$ and $X=-1$ is transmitted,
respectively. Let us now flip the density associated with $-1$, i.e.,
we consider $\Ldens{a}^-(-y)$, so that positive values indicate ``correct''
messages. By the symmetry of the message-passing equations (see
Definition 4.81 in \cite{RiU08}), the sign of all those messages
that enter or exit the variable nodes with associated transmitted
value $-1$ is flipped as well. Therefore, the density evolution
analysis for a particular codeword is equivalent to that for the
all-$1$ codeword, provided that we initialize the variable nodes
with associated value $1$ and $-1$ to $\Ldens{a}^+(y)$ and
$\Ldens{a}^-(-y)$, respectively.  Now, each transmitted bit is
independent and uniformly distributed. Thus, we pick a variable
node with $L$-density $\Ldens{a}^+(y)$ with probability $1/2$ and
with $L$-density $\Ldens{a}^-(-y)$ with probability $1/2$. As a
result, the density evolution equations for our asymmetric setting
are the same as those for transmission over the ``symmetrized
channel'' with $L$-density given by
\begin{equation}\label{eq:ldenssymm}
\Ldens{a}^{\rm s}(y) = \frac{1}{2}(\Ldens{a}^-(-y)+\Ldens{a}^+(y)).
\end{equation}
This channel is, indeed, symmetric and its capacity equals the
symmetric capacity of the actual channel $W$ over which transmission
takes place. These two results are formalized by the propositions
below that are proved in Appendix \ref{app:proofsecII1}.

\begin{proposition}\label{prop:symm}
Consider the transmission over a B-DMC and let $\Ldens{a}^+(y)$ and $\Ldens{a}^-(y)$ be the
$L$-densities assuming that $X=1$ and $X=-1$ is transmitted, respectively. Then, the $L$-density $\Ldens{a}^{\rm s}(y)$ given by \eqref{eq:ldenssymm} is symmetric.
\end{proposition}

\begin{proposition}\label{prop:cap}
Consider the transmission over a B-DMC $W$ with symmetric capacity $C_{\rm s}(W)$ and let $\Ldens{a}^+(y)$ and $\Ldens{a}^-(y)$ be the
$L$-densities assuming that $X=1$ and $X=-1$ is transmitted, respectively. Define the $L$-density of the ``symmetrized channel'' $\Ldens{a}^{\rm s}(y)$ as in \eqref{eq:ldenssymm}. Then,
\begin{equation}
C_{\rm s}(W) = \int \Ldens{a}^{\rm s}(y)\left(1-\log(1+e^{-y})\right)dy.
\end{equation}
\end{proposition}

Consequently, in order to achieve the symmetric capacity $C_{\rm s}(W)$ of the (possibly asymmetric) channel $W$, it suffices to
construct a code that achieves the capacity of the symmetric channel
with $L$-density $\Ldens{a}^{\rm s}$. Indeed, the density evolution
analysis for the transmission over $W$ is exactly the same as for the
transmission over the symmetrized channel. Furthermore, the capacity
of the symmetrized channel equals the symmetric capacity $C_{\rm s}(W)$ of the original channel $W$. As a result, in order to solve
the problem, we can employ, for instance, an $(\dl, \dr)$-regular
SC-LDPC ensemble with sufficiently large degrees.

In short, the problem of achieving the symmetric capacity of any B-DMC can be solved by using codes (e.g., polar, spatially coupled) that are provably optimal for symmetric channels.

\subsection{How to Transmit Biased Bits} \label{subsec:checks}

Let us consider a generalization of the previous problem in which
the bits of the codeword are biased, i.e., they are not chosen according to a uniform
distribution. This scenario is an important primitive that will be used in Sections \ref{sec:integrated} and \ref{sec:chaining}, where we describe coding techniques that achieve the capacity of asymmetric channels.

\vspace{1em}

\noindent {\bf Problem Statement.} Let $W$ be a B-DMC with capacity-achieving input distribution $\{p^*(x)\}_{x\in \{0, 1\}}$ s.t. $p^*(1)=\alpha$ for some $\alpha \in [0, 1]$. Let $X^{1:n}$ be a sequence of $n$
i.i.d. Bernoulli$(\alpha)$ random variables. Denote by $Y^{1:n}$ the channel output when $X^{1:n}$ is transmitted. Furthermore, assume that the transmitter and the receiver are connected via a noiseless channel of capacity roughly
$nH(X\mid Y)$. Given $Y^{1:n}$, and with the help of the noiseless channel, the aim is to reconstruct $X^{1:n}$ at the receiver with high probability as $n$ goes large. 

\vspace{1em}

\noindent {\bf Design of the Scheme.} Let $P$ be a parity-check matrix with $nH(X\mid Y)$ rows and $n$ columns. Hence, $P$ represents a code of length $n$ and rate $1-H(X\mid Y)$. Let $S^{1:nH(X\mid Y)} = P X^{1:n}$ and assume that $S^{1:nH(X\mid Y)}$ is sent over the noiseless channel to the receiver. Given $Y^{1:n}$ and $S^{1:nH(X\mid Y)}$, we will prove that the receiver can reconstruct $X^{1:n}$, assuming that the code represented by $P$ is capacity-achieving under belief-propagation (BP) decoding for the symmetric channel described in the following (see eq. \eqref{eq:ldenssymmpost}). 

For the sake of simplicity, let us map the input alphabet into $\{-1, 1\}$. Since the input distribution is non-uniform, the belief-propagation algorithm needs to take into account also the prior on $X$ and it is no longer based on the log-likelihood ratio $L(y)$ defined in  \eqref{eq:loglike}. Let $L_{\rm p}(y)$ denote the log-posterior ratio, defined as 
\begin{equation}\label{eq:loglikepost}
L_{\rm p}(y) = \ln \frac{p_{X\mid Y}(1\mid y)}{p_{X\mid Y}(-1\mid y)} = L(y) + \ln \frac{\bar{\alpha}}{\alpha}.
\end{equation}
Following the lead of Section \ref{subsec:symmcap}, let us define the densities of $L_{\rm p}(Y)$ assuming that $X=1$ and $X=-1$ is transmitted and let us denote them as $\Ldens{a}_{\rm p}^+(y)$ and $\Ldens{a}_{\rm p}^-(y)$, respectively. If we flip the density associated with $X=-1$, i.e., we consider $\Ldens{a}_{\rm p}^-(-y)$, then, by the symmetry of the message-passing equations, the sign of the messages that enter or exit the variable nodes with associated transmitted value $X=-1$ is flipped as well. Therefore, the density evolution analysis for a particular codeword is equivalent to that for the all-one codeword provided that we initialize the variable nodes with associated value $1$ to $\Ldens{a}_{\rm p}^+(y)$,  and the variable nodes with value $-1$ to $\Ldens{a}_{\rm p}^-(-y)$, respectively. As ${\mathbb P}(X=-1) = \alpha$, the density evolution equations for our asymmetric setting are the same as those for transmission over the ``symmetrized channel'' with $L$-density
\begin{equation}\label{eq:ldenssymmpost}
\Ldens{a}_{\rm p}^{\rm s}(y) = \alpha \Ldens{a}_{\rm p}^-(-y) + \bar{\alpha}\Ldens{a}_{\rm p}^+(y).
\end{equation}
The propositions below show that this channel is, indeed, symmetric and establish the relation between the conditional entropy $H(X\mid Y)$ and $\Ldens{a}_{\rm p}^{\rm s}(y)$. The proofs of these results can be found in Appendix \ref{app:proofsecII2}. 

\begin{proposition}\label{prop:symmalpha}
The $L$-density $\Ldens{a}_{\rm p}^{\rm s}(y)$ given by \eqref{eq:ldenssymmpost} is symmetric.
\end{proposition}

\begin{proposition}\label{prop:capalpha}
Consider the transmission over a B-DMC $W$ with capacity-achieving input distribution $p^*$. Let $X\sim p^*$ and $Y$ be the input and the output of the channel. Denote by $\Ldens{a}_{\rm p}^+(y)$ and $\Ldens{a}_{\rm p}^-(y)$ the densities of $L_{\rm p}(Y)$ assuming that $X=1$ and $X=-1$ is transmitted. Let $\Ldens{a}_{\rm p}^{\rm s}(y)$ be the density of the ``symmetrized channel'', as in \eqref{eq:ldenssymmpost}. Then,
\begin{equation}\label{eq:condent}
H(X\mid Y) = \int \Ldens{a}_{\rm p}^{\rm s}(y)\log(1+e^{-y})dy.
\end{equation}
\end{proposition}

Let us now see how these two propositions imply that we can reconstruct $X^{1:n}$ with high probability. Since the channel with density $\Ldens{a}_{\rm p}^{\rm s}(y)$ is symmetric by Proposition \ref{prop:symmalpha}, its capacity is given by
\begin{equation*}
\int \Ldens{a}_{\rm p}^{\rm s}(y)\left(1-\log(1+e^{-y})\right)dy = 1-\int \Ldens{a}_{\rm p}^{\rm s}(y)\log(1+e^{-y})dy = 1-H(X\mid Y),
\end{equation*}
where the last equality comes from Proposition \ref{prop:capalpha}.
Recall that the receiver is given the channel output $Y^{1:n}$ and
the error-free vector $S^{1:nH(X\mid Y)} = P X^{1:n}$. Hence, we
can think of this setting as one in which $X^{1:n}$ is a codeword
of a sparse graph code with syndrome vector $S^{1:nH(X\mid Y)}$
shared between the transmitted and the receiver. As previously
stated, the density evolution analysis for this case is the same
as when we transmit over the symmetric channel with density
$\Ldens{a}_{\rm p}^{\rm s}(y)$ and the syndrome vector is set to
0. By assumption, the matrix $P$ comes from a code that achieves
capacity for such a symmetric channel, hence the transmitted vector
$X^{1:n}$ can be reconstructed with high probability.

We can employ, for instance, an $(\dl, \dr)$-regular SC-LDPC ensemble
with sufficiently large degrees. Another option is to use
spatially coupled MacKay-Neal (MN) and Hsu-Anastasopoulos (HA) LDPC
codes that, compared to $(\dl, \dr)$-regular codes, have bounded
graph density. In particular, in \cite{KaSa11} it is proved that
MN and HA codes achieve the capacity of B-DMCs \emph{under MAP
decoding} by using a parity-check matrix with bounded column and row
weight.  Furthermore, the authors of \cite{KaSa11} give empirical
evidence of the fact that spatially coupled MN and HA codes achieve
the capacity of the BEC also {\em under iterative decoding}. These results
are extended to the additive white Gaussian noise channel in
\cite{MKLC12}.

In summary, so far we have discussed how to achieve the symmetric
capacity of a B-DMC and how to transmit biased bits. Now, let us
move to the main topic of this paper and describe three approaches
for achieving the actual capacity of any DMC. While doing so, we will regard the solutions to the two problems of this section as useful primitives.

\section{Paradigm 1: Gallager's Mapping}\label{sec:gallager}

The solution proposed by Gallager \cite[pg. 208]{Gal68} consists
of using a standard linear code and applying a non-linear mapper
to the encoded bits in such a way that the resulting input distribution
is appropriately biased. More recently, Gallager's mapping was used
in \cite{McE01} to approach the capacity of nonstandard channels
via turbo-like codes. Furthermore, in \cite{SS03allerton}, the
authors applied this shaping idea to finite-state channels and described
how to construct an explicit invertible finite-state encoder. Before
moving on to a general description of the scheme, to convey the
main ideas, let us start with an example.

\subsection{A Concrete Example}\label{subsec:ex}
Let ${\mathcal X} = \{0, 1, 2\}$ and suppose that we want to transmit
over a channel $W: \mathcal {X} \to {\mathcal Y}$ with a
capacity-achieving input distribution of the following form:
$p^*(0)=3/8$, $p^*(1)=3/8$, $p^*(2)=2/8$.

Let ${\mathcal V} = \{0, 1, \cdots, 7\}$ and
consider the function $f : {\mathcal V} \to {\mathcal X}$ that
maps three elements of ${\mathcal V}$ into $0\in {\mathcal X}$, three
other elements of ${\mathcal V}$ into $1\in {\mathcal X}$, and the remaining two elements
of ${\mathcal V}$ into $2\in {\mathcal X}$. In this way, the uniform
distribution over ${\mathcal V}$ induces the capacity achieving
distribution over ${\mathcal X}$. Define the channel $W': {\mathcal
V}\to {\mathcal Y}$ as
\begin{equation}\label{eq:mapqex}
W'(y \mid v) = W(y \mid f(v) ).
\end{equation}
Take a code that achieves the symmetric capacity of $W'$. Then, we can use this code to achieve the capacity of $W$ via the mapping $f$.

The above scheme works under the assumption that we can construct
codes that achieve the symmetric capacity for any given input
alphabet size. Note that this can be done, e.g., with $q$-ary polar
codes \cite{STA09}. Sometimes it is more convenient to achieve this
goal indirectly by using only binary codes. Indeed, suppose that
the channel changes for some reason. Then, the optimal input
distribution also changes, and we might have to change the alphabet
${\mathcal V}$. If we code directly on ${\mathcal V}$, we will also
have to change the code itself. If the code needs to be implemented
in hardware, this might not be convenient. However, if we manage
to use the same binary code and only modify some preprocessing
steps, then it is easy to accomplish any required change in the
input distribution.

Let us now describe this approach in detail. Observe that ${\mathcal V}$ has cardinality $8=2^3$. Rather than considering the set of integers from $0$ to $7$, it is more convenient to consider the set of binary triplets. Let ${\mathcal U} = \{0, 1\}^3$ and consider the function by
$g : {\mathcal U} \to {\mathcal X}$. As before, $g$ maps three elements of
${\mathcal U}$ into $0\in {\mathcal X}$, three other elements of ${\mathcal
U}$ into $1\in {\mathcal X}$, and the remaining two elements of ${\mathcal U}$ into
$2\in {\mathcal X}$. In this way, the uniform
distribution over ${\mathcal U}$ induces the capacity achieving
distribution over ${\mathcal X}$. Note that any $u\in {\mathcal U}$ can be written as $u=(u_1, u_2, u_3)$, where
$u_i \in \{0, 1\}$ for $i \in \{1, 2, 3\}$. Define the channels $W''_{1}: \{0, 1\}\to
{\mathcal Y}$, $W''_{2}: \{0, 1\}\to {\mathcal Y}\times \{0, 1\}$,
$W''_{3}: \{0, 1\}\to {\mathcal Y}\times \{0, 1\} \times \{0,
1\}$ as
\begin{equation}\label{eq:mapbinex}
\begin{split}
&W''_{1}(y \mid u_1) =\frac{1}{4} \sum_{u_2, u_3} W(y \mid g(u_1, u_2, u_3) ),\\
&W''_{2}(y, u_1 \mid u_2) =\frac{1}{4} \sum_{u_3} W(y \mid g(u_1, u_2, u_3) ),\\
&W''_{3}(y, u_1, u_2 \mid u_3) =\frac{1}{4} W(y \mid g(u_1, u_2, u_3) ).\\
\end{split}
\end{equation}
Take three binary codes that achieve the symmetric capacities of $W''_{1}$, $W''_{2}$, and $W''_{3}$. By the chain rule of mutual information, the sum of these capacities equals $C(W)$. Hence, we can use these codes to achieve the capacity of $W$ via the mapping $g$.

\subsection{Description of the General Scheme}\label{subsec:generalscheme}

\noindent {\bf Problem Statement.} Let $W$ be a DMC with capacity-achieving input distribution $\{p^*(x)\}_{x\in {\mathcal X}}$. The aim is to transmit over $W$ with rate close to $C(W)$.

\vspace{1em}

\noindent {\bf Design of the Scheme.}
Pick $\delta >0$ and find a rational approximation $\tilde{p}(x)$ that differs from $p^*(x)$ by at most $\delta$ in total variation distance. In formulae, take $\tilde{p}(x) = n_x/d$ with $n_x$, $d \in {\mathbb N}$ for all $x\in {\mathcal X}$ s.t.
\begin{equation}\label{eq:TVdist}
\frac{1}{2}\sum_{x\in {\mathcal X}}|p^*(x)-\tilde{p}(x)| < \delta.
\end{equation}
Take an extended alphabet ${\mathcal V}$ with cardinality equal to $d$ and consider the function $f : {\mathcal V} \to {\mathcal X}$ that maps $n_x$ elements of ${\mathcal V}$ into $x\in {\mathcal X}$. Define the channel $W': {\mathcal V}\to {\mathcal Y}$ as in \eqref{eq:mapqex}. Denote by $X$ and $Y$ the input and the output of the channel $W$, respectively. Let $V$ be uniformly distributed over $\mathcal V$ and set $X = f(V)$. Since the uniform distribution over ${\mathcal V}$ induces the input distribution $\tilde{p}(x)$ over $\mathcal X$, we have that $X\sim \tilde{p}(x)$. Construct a code $\mathcal C$ that achieves the symmetric capacity of $W'$. Therefore, by using the code $\mathcal C$ and the mapping $f$, we can transmit at rate $R$ arbitrarily close to 
\begin{equation*}
C_{\rm s}(W') = I(V; Y) = I(X; Y) \underset{\delta\to 0}{\longrightarrow} C(W).
\end{equation*}
As $\delta$ goes to 0, the distribution $\tilde{p}$ tends to $p^*$ and $I(X; Y)$ approaches $C(W)$.

If we want to restrict to binary codes, select a rational approximation of the form $\tilde{p}(x) = n_x/2^t$ for $t, n_x \in {\mathbb N}$. Pick ${\mathcal U} = \{0, 1\}^t$ and consider the function $g : {\mathcal U} \to {\mathcal X}$ that maps $n_x$ elements of ${\mathcal U}$ into $x\in {\mathcal X}$. The set ${\mathcal U}$ contains binary vectors of length $t$ that can be written in the form $u_{1:t}=(u_1, \cdots, u_t)^T$, where $u_j \in \{0, 1\}$ for $j \in [t]$. Define the synthetic channels $W''_{j}: \{0, 1\}\to {\mathcal Y}\times \{0, 1\}^{j-1}$, similarly to \eqref{eq:mapbinex}, i.e.,  
\begin{equation}\label{eq:defWg}
W''_{j}(y, u_{1:j-1} \mid u_j) =\frac{1}{2^{t-1}} \sum_{u_{j+1:t}} W(y \mid g(u_{1:t}) ).
\end{equation}
Let $U_{1:t}$ be a sequence of $t$ i.i.d. random variables uniform over $\{0, 1\}$. Set $X = g(U_{1:t})$. Since the uniform distribution over ${\mathcal U}$ induces the input distribution $\tilde{p}(x)$ on $\mathcal X$, we have that $X\sim \tilde{p}(x)$. Construct $t$ codes ${\mathcal C}_{1}, \cdots, {\mathcal C}_{t}$ s.t. ${\mathcal C}_{j}$ has rate $R_j$ that is arbitrarily close to the symmetric capacity of the channel $W''_{j}$. Therefore, by using these codes and the mapping $g$, we can transmit at rate $R$ arbitrarily close to
\begin{equation}\label{eq:chainrulegall}
\sum_{j = 1}^t I(U_j; Y \mid U_{1:j-1}) = I (X; Y) \underset{t\to\infty}{\longrightarrow} C(W),
\end{equation}
where the first equality comes from the chain rule and $I(X;Y)$ approaches $C(W)$ as $\delta$ goes to 0.

Another way to interpret this problem is from the perspective of coding for the $t$-user multiple access channel (MAC), see Section V-A of \cite{AbT12}. The paper \cite{AbT12} focuses on polar codes, but the same considerations apply to any coding scheme that achieves the uniform sum-rate of a $t$-user MAC. Indeed, let $W''$ be a DMC whose input alphabet has cardinality $2^t$ and consider the channel $\widetilde{W}''$ obtained by mapping a uniform symbol on the input alphabet of $W''$ into $t$ independent uniform bits. Then, $\widetilde{W}''$ is a $t$-user binary-input MAC and its uniform sum-rate equals the symmetric capacity of $W''$. Thus, in order to achieve the symmetric capacity of $W''$ (hence, the capacity of the original channel $W$), we can use a code for the MAC $\widetilde{W}''$.     

Let us now explain formally how the encoding and decoding operations are done for the schemes mentioned above (see also Figure \ref{fig:gallscheme}). Then, we will consider the performance of this approach by relating the gap $C(W)-I(X;Y)$ to $\delta$ and to the cardinalities of the input and the output alphabets.

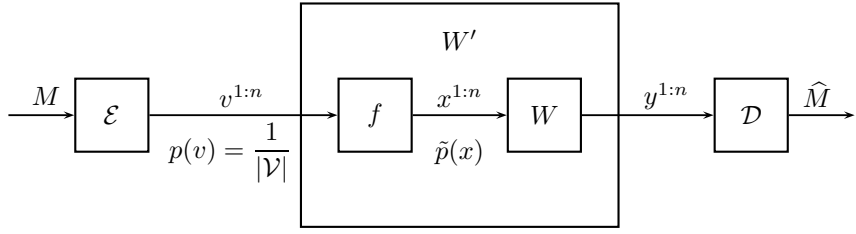
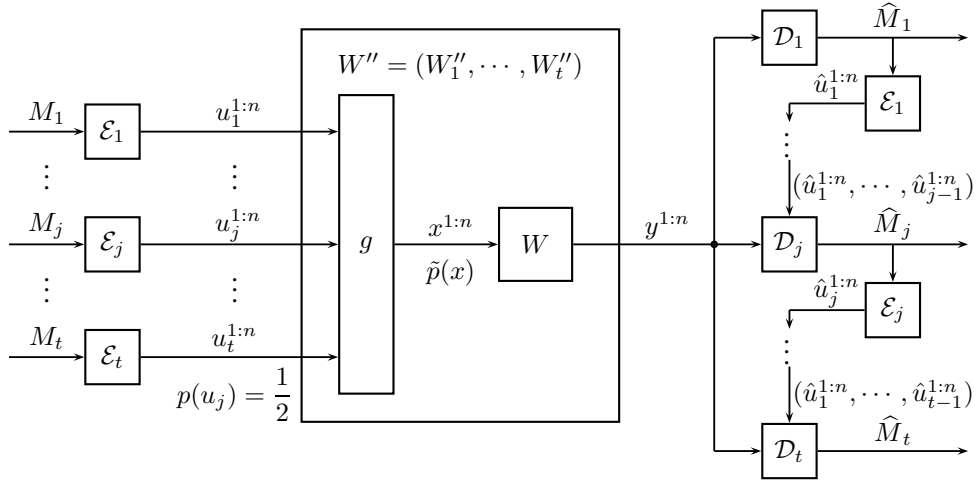
\begin{figure}[t]
    \centering
    \subfigure[Solution based on a single non-binary code: the message $M$ is encoded by $\mathcal E$ and decoded by $\mathcal D$.]{
\psset{arrowscale=1}
\psset{unit=0.5cm}
\psset{xunit=1,yunit=1}
\begin{pspicture}(-6,0.5)(19,7)
\psline[linecolor=black,linewidth=0.7pt]{->}(-1.75,4)(3.25,4)
\psline[linecolor=black,linewidth=0.7pt]{->}(-5.5,4)(-3.75,4)
\psframe(2.25,1)(10.75,7)
\psframe(3.25,3)(5.25,5)
\psframe(-3.75,3)(-1.75,5)
\rput(-2.75,4){\small{$\mathcal E$}}
\psframe(13.25,3)(15.25,5)
\rput(14.25,4){\small{$\mathcal D$}}
\psline[linecolor=black,linewidth=0.7pt]{->}(5.25,4)(7.75,4)
\psline[linecolor=black,linewidth=0.7pt]{->}(15.25,4)(17,4)
\rput(4.25,4){\small{$f$}}
\psframe(7.75,3)(9.75,5)
\rput(8.75,4){\small{$W$}}
\psline[linecolor=black,linewidth=0.7pt]{->}(9.75,4)(13.25,4)
\rput[c](0.7,4.5){\small{$v^{1:n}$}}
\rput[c](-4.5,4.5){\small{$M$}}
\rput[c](0.4,3){\small{$p(v) = \displaystyle\frac{1}{|{\mathcal V}|}$}}
\rput[c](6.5,3){\small{$\tilde{p}(x)$}}
\rput[c](6.5,4.5){\small{$x^{1:n}$}}
\rput[c](12,4.5){\small{$y^{1:n}$}}
\rput[c](16,4.5){\small{$\fixwidehat{M}$}}
\rput[c](6.5,6){\small{$W'$}}
\end{pspicture}
}
    \subfigure[Solution based on $t$ binary codes: the message $M=(M_1, \cdots, M_t)$ is encoded by ${\mathcal E}_1, \cdots, {\mathcal E}_t$ and decoded successively by ${\mathcal D}_1, \cdots, {\mathcal D}_t$. Note that ${\mathcal D}_j$ is the decoder of the synthetic channel $W''_{j}$ and it is fed with the output of the actual channel $W$ together with the previous re-encoded estimates.]{
\psset{arrowscale=1}
\psset{unit=0.5cm}
\psset{xunit=1,yunit=1}
\begin{pspicture}(-10,-5)(26,8)
\psline[linecolor=black,linewidth=0.7pt]{->}(-2,4.75)(3.25,4.75)
\psline[linecolor=black,linewidth=0.7pt]{->}(-2,1.75)(3.25,1.75)
\psline[linecolor=black,linewidth=0.7pt]{->}(-2,-1.25)(3.25,-1.25)

\psline[linecolor=black,linewidth=0.7pt]{->}(-5.5,4.75)(-3.5,4.75)
\psline[linecolor=black,linewidth=0.7pt]{->}(-5.5,1.75)(-3.5,1.75)
\psline[linecolor=black,linewidth=0.7pt]{->}(-5.5,-1.25)(-3.5,-1.25)

\rput[c](-4.5,5.25){\small{$M_1$}}
\rput[c](-4.5,2.25){\small{$M_j$}}
\rput[c](-4.5,-0.75){\small{$M_t$}}

\psframe(-3.5,4)(-2,5.5)
\rput(-2.75,4.75){\small{${\mathcal E}_1$}}

\psframe(-3.5,1)(-2,2.5)
\rput(-2.75,1.75){\small{${\mathcal E}_j$}}

\psframe(-3.5,-2)(-2,-0.5)
\rput(-2.75,-1.25){\small{${\mathcal E}_t$}}

\psframe(2.25,-3)(10.75,7.5)
\rput(0.5,3.75){$\vdots$}
\rput(0.5,0.75){$\vdots$}
\rput(-4.5,3.75){$\vdots$}
\rput(-4.5,0.75){$\vdots$}
\rput[c](0.6,5.25){\small{$u_1^{1:n}$}}
\rput[c](0.6,2.25){\small{$u_j^{1:n}$}}
\rput[c](0.5,-0.75){\small{$u_t^{1:n}$}}

\rput[c](0.5,-2.25){\small{$p(u_j) = \displaystyle\frac{1}{2}$}}

\rput[c](6.25,2.25){\small{$x^{1:n}$}}
\rput[c](6.25,1){\small{$\tilde{p}(x)$}}

\psframe(3.25,-2.25)(4.75,5.75)
\rput[c](4,1.75){\small{$g$}}
\psline[linecolor=black,linewidth=0.7pt]{->}(4.75,1.75)(7.5,1.75)
\psframe(7.5,0.75)(9.5,2.75)
\rput(8.5,1.75){\small{$W$}}

\pscircle(13.25, 1.75){0.1}

\psline[linecolor=black,linewidth=0.7pt]{->}(9.5,1.75)(14.5,1.75)
\rput[c](12,2.25){\small{$y^{1:n}$}}

\psline[linecolor=black,linewidth=0.7pt]{->}(16,7.25)(20,7.25)
\rput(18,7.75){\small{$\fixwidehat{M}_1$}}

\psline[linecolor=black,linewidth=0.7pt]{->}(16,1.75)(20,1.75)
\rput(18,2.25){\small{$\fixwidehat{M}_j$}}

\psline[linecolor=black,linewidth=0.7pt]{->}(16,-3.75)(20,-3.75)
\rput(18,-3.25){\small{$\fixwidehat{M}_t$}}

\psframe(14.5,6.5)(16,8)
\rput(15.25,7.25){\small{${\mathcal D}_1$}}

\psframe(14.5,1)(16,2.5)
\rput(15.25,1.75){\small{${\mathcal D}_j$}}

\psframe(14.5,-4.5)(16,-3)
\rput(15.25,-3.75){\small{${\mathcal D}_t$}}

\psline[linecolor=black,linewidth=0.7pt]{-}(13.25,7.25)(13.25,-3.75)
\psline[linecolor=black,linewidth=0.7pt]{->}(13.25,7.25)(14.5,7.25)
\psline[linecolor=black,linewidth=0.7pt]{->}(13.25,-3.75)(14.5,-3.75)

\psline[linecolor=black,linewidth=0.7pt]{->}(18,7.25)(18,6.25)
\psframe(17.25,6.25)(18.75,4.75)
\rput(18,5.5){\small{${\mathcal E}_1$}}
\psline[linecolor=black,linewidth=0.7pt]{-}(17.25,5.5)(15.25,5.5)
\rput(16.5,6){\small{$\hat{u}_1^{1:n}$}}
\psline[linecolor=black,linewidth=0.7pt]{->}(15.25,5.5)(15.25,5)
\rput(15.175,4.6){$\vdots$}
\psline[linecolor=black,linewidth=0.7pt]{->}(15.25,4)(15.25,2.5)
\rput(17.75,3.25){\small{$(\hat{u}_1^{1:n}, \cdots, \hat{u}_{j-1}^{1:n})$}}

\psline[linecolor=black,linewidth=0.7pt]{->}(18,1.75)(18,0.75)
\psframe(17.25,0.75)(18.75,-0.75)
\rput(18,0){\small{${\mathcal E}_j$}}
\psline[linecolor=black,linewidth=0.7pt]{-}(17.25,0)(15.25,0)
\rput(16.5,0.5){\small{$\hat{u}_j^{1:n}$}}
\psline[linecolor=black,linewidth=0.7pt]{->}(15.25,0)(15.25,-0.5)
\rput(15.175,-0.9){$\vdots$}
\psline[linecolor=black,linewidth=0.7pt]{->}(15.25,-1.5)(15.25,-3)
\rput(17.75,-2.25){\small{$(\hat{u}_1^{1:n}, \cdots, \hat{u}_{t-1}^{1:n})$}}

\rput[c](6.5,6.5){\small{$W'' = (W_{1}'', \cdots, W_{t}'')$}}
\end{pspicture}
}

\caption{Coding over asymmetric channels via Gallager's mapping.}
\label{fig:gallscheme}
\end{figure}

\vspace{1em}

\noindent {\bf Encoding.} First, consider the scheme based on a single non-binary code. Let $M$ be the information message that can be thought of as a binary string of length $nR$ and let $\mathcal E$ be the encoder of the code $\mathcal C$. The output of the encoder is $v^{1:n} = (v^{(1)}, \cdots, v^{(n)})^T$, where $v^{(i)} \in {\mathcal V}$ for $i \in \{1, \cdots, n\}$. Then, $v^{1:n}$ is mapped component-wise by the function $f$ into $x^{1:n} = (x^{(1)}, \cdots, x^{(n)})^T$, with $x^{(i)} \in {\mathcal X}$ s.t. $x^{(i)} = f(v^{(i)})$. 

Second, consider the scheme based on $t$ binary codes. Let $M = (M_1, \cdots, M_t)$ be the information message divided into $t$ parts so that $M_j$ can be thought of as a binary string of length $nR_j$ for $j \in \{1, \cdots, t\}$. Let ${\mathcal E}_j$ be the encoder of the code ${\mathcal C}_j$ that maps $M_j$ into $u_j^{1:n} = (u_j^{(1)}, \cdots, u_j^{(n)})^T$, where $u_j^{(i)}\in \{0, 1\}$ for $i\in\{1, \cdots, n\}$. Then, $u_{1:t}^{1:n}$ is mapped component-wise by the function $g$ into $x^{1:n} = (x^{(1)}, \cdots, x^{(n)})^T$, with $x^{(i)} \in {\mathcal X}$ s.t. $x^{(i)} = g(u_{1:t}^{(i)})$.

Finally, we transmit the sequence $x^{1:n}$ over the channel $W$.

\vspace{1em}

\noindent {\bf Decoding.} First, consider the scheme based on a single non-binary code. Let $\mathcal D$ be the decoder of the code $\mathcal C$, that accepts as input the channel output $y^{1:n}$ and outputs the estimate $\fixwidehat{M}$.

Second, consider the scheme based on $t$ binary codes. Let ${\mathcal D}_j$ be the decoder of the code ${\mathcal C}_j$. It accepts as input the channel output $y^{1:n}$ and the previous re-encoded estimates $(\fixwidehat{u}_1^{1:n}, \cdots, \fixwidehat{u}_{j-1}^{1:n})$. It outputs the current estimate $\fixwidehat{M}_j$. To make the use of the previous estimates possible, the decoding occurs successively, i.e., the decoders ${\mathcal D}_1, \cdots, {\mathcal D}_t$ are activated in series. 

The situation is schematically represented in Figure \ref{fig:gallscheme}.

\vspace{1em}

\noindent {\bf Performance.} The codes $\mathcal C$ and ${\mathcal C}_j$ can be used to transmit reliably at rates $R$ and $R_j$. Then, $\fixwidehat{M} = M$ and $\fixwidehat{M}_j = M_j$ ($j \in \{1, \cdots t\}$) with high probability. As a result, we can transmit over $W$ with rate close to $I(X;Y)$, where the
input distribution is $\tilde{p}(x)$. Also, since the mutual information
is a continuous function of the input distribution, if $\delta$ gets
small, then $I(X;Y)$ approaches $C(W)$. This statement is
made precise by the following proposition that is proved in
Appendix~\ref{app:boundmi}.

\begin{proposition}\label{prop:boundmi}
Consider the transmission over the channel $W: {\mathcal X} \to {\mathcal
Y}$ and let $I(p)$ be the mutual information between the input and
the output of the channel when  the input distribution is $p$. Let
$p$ and $p^*$ be input distributions such that their total variation distance
is upper bounded by $\delta$, as in \eqref{eq:TVdist}, for $\delta \in (0, 1/8)$. Then,
\begin{equation}\label{eq:boundY}
|I(p^*)-I(p)| < 3\delta \log |{\mathcal Y}| +h_2(\delta),
\end{equation}
\begin{equation}\label{eq:boundX}
|I(p^*)-I(p)| < 7\delta \log |{\mathcal X}| +h_2(\delta) + h_2(4\delta).
\end{equation}
\end{proposition}

Note that the bounds \eqref{eq:boundY} and \eqref{eq:boundX} depend separately on the input and the output alphabet. Therefore, we can conclude that, under the hypotheses of Proposition \ref{prop:boundmi}, 
\begin{equation*}
|I(p^*)-I(p)| = O\left(\delta\log\left(\frac{\min(|{\mathcal X}|, |{\mathcal Y}|)}{\delta}\right)\right).
\end{equation*}

\section{Paradigm 2: Integrated Scheme} \label{sec:integrated}

The basic idea of this approach is to use a coding scheme that is
simultaneously good for lossless source coding and for channel
coding. The \emph{source coding} part is needed to create a biased
input distribution from uniform bits, whereas the \emph{channel coding} part provides reliability for the transmission over the channel.

A provably capacity-achieving scheme was first proposed in
\cite{HY13} in the context of polar codes (see also Chapter 3 of \cite{Hon13}). Such a scheme is reviewed in Section \ref{subsec:polar} and the interested reader is referred to \cite{MHSU14} for a more detailed illustration. Then, in Section \ref{subsec:sparse} we describe how to extend the idea to sparse graph codes. A similar scheme was proposed in \cite{MuMi10} under MAP decoding. On the contrary, our approach employs a low-complexity belief-propagation algorithm.

\vspace{1em}

\noindent {\bf Problem Statement.} Let $W$ be a B-DMC with capacity-achieving input distribution $\{p^*(x)\}_{x\in \{0, 1\}}$ s.t. $p^*(1)=\alpha$ for some $\alpha \in [0, 1]$. The aim is to transmit over $W$ with rate close to $C(W)$.

\subsection{Polar Codes -- \cite{HY13}} \label{subsec:polar}

\noindent {\bf Design of the Scheme.}
Let $(U^{1:n})^T = (X^{1:n})^T G_n$, where $X^{1:n}$ is a column vector of $n$
i.i.d. components drawn according to the capacity-achieving input distribution $p^*$.

Let us start with the \emph{source coding} part of the scheme. Consider the sets
$\hset_X$ and $\lset_X$ defined as follows: for $i \in \hset_X$, $U^{(i)}$ is approximately uniformly
distributed and independent of $U^{1:i-1}$; for $i \in \lset_X$,
$U^{(i)}$ is approximately a deterministic function of the past
$U^{1:i-1}$. In formulae,
\begin{equation} \label{lossless_sets}
\begin{split}
\hset_X &= \{i \in [n] : Z(U^{(i)} \mid  U^{1:i-1}) \ge 1- \delta_n\}, \\
\lset_X &= \{i \in [n] : Z(U^{(i)} \mid  U^{1:i-1}) \le \delta_n\},
\end{split}
\end{equation}
where $\delta_n= 2^{-n^{\beta}}$ for $\beta \in (0, 1/2)$ and $Z(\cdot \mid \cdot)$ denotes the Bhattacharyya parameter. More formally, given
$(T, V)\sim p_{T, V}$, where $T$ is binary and $V$ takes values in
an arbitrary discrete alphabet ${\mathcal V}$, we define
\begin{equation} \label{eq:Bhattacharyya}
Z(T\mid V) = 2 \sum_{v\in{\mathcal V}}
{\mathbb P}_V(v)\sqrt{{\mathbb P}_{T\mid V}(0\mid v){\mathbb P}_{T\mid V}(1\mid v)}.
\end{equation}
It can be shown that the Bhattacharyya parameter $Z(T\mid V)$ is close to 0 or 1 if and only if the conditional entropy $H(T\mid V)$ is 
close to 0 or 1 (see Proposition 2 of \cite{Ar10}). Consequently, if $Z(T\mid V)$ is close to zero, then $T$ is approximately a deterministic function of $V$, and, 
if $Z(T\mid V)$ is close to $1$, then $T$ is approximately uniformly
distributed and independent of $V$.
Furthermore, 
\begin{equation} \label{eq:card}
\begin{split}
\lim_{n\to \infty} \frac{1}{n} \, | \hset_X|  &= H(X),\\
\lim_{n\to \infty}\frac{1}{n} \, | \lset_X|  &= 1-H(X).\\
\end{split}
\end{equation}
This means that with high probability either $U^{(i)}$ is independent from $U^{1:i-1}$ or is a deterministic function of it (for a detailed proof of polarization, see \cite{Ari09}).

Let $\lset_X^{\rm c}$ be the complement of $\lset_X$, and note that $\lim_{n\to \infty} | \lset_X^{\rm c}|/n  = H(X)$. Furthermore, given $\{U^{(i)}\}_{i \in \lset_X^{\rm c}}$, we can recover the whole vector $U^{1:n}$ in a successive manner, since $U^{(i)}$ is approximately a deterministic function of $U^{1:i-1}$ for $i \in \lset_X$. Consequently, we can compress $X^{1:n}$ into the sequence $\{U^{(i)}\}_{i \in \lset_X^{\rm c}}$ that has size roughly $nH(X)$. 

Now, let us consider the \emph{channel coding} part of the scheme. Assume that the channel output $Y^{1:n}$ is given, and interpret
this as side information for $X^{1:n}$. Consider the sets $\hset_{X
\mid Y}$ and $\lset_{X \mid Y}$ that contain, respectively, the positions $i \in [n]$ in which
$U^{(i)}$ is approximately uniform and independent of
$(U^{1:i-1}, Y^{1:n})$ and approximately a deterministic function
of $(U^{1:i-1}, Y^{1:n})$. In formulae,
\begin{equation}\label{eq:cardint1}
\begin{split}
\hset_{X\mid Y} &= \{i \in [n] : Z(U^{(i)} \mid U^{1:i-1}, Y^{1:n}) \ge 1- \delta_n\},\\
\lset_{X\mid Y} &= \{i \in
[n] : Z(U^{(i)} \mid  U^{1:i-1}, Y^{1:n}) \le \delta_n\}.\\
\end{split}
\end{equation}
Analogously to \eqref{eq:card}, we have that
\begin{equation}\label{eq:cardint2}
\begin{split}
\lim_{n\to \infty} \frac{1}{n} \, | \hset_{X\mid Y}|  &= H(X\mid Y),\\
\lim_{n\to \infty}\frac{1}{n} \, | \lset_{X\mid Y}|  &= 1-H(X\mid Y).\\
\end{split}
\end{equation}
Furthermore, analogously to the case above, given the channel output $Y^{1:n}$ and the sequence $\{U^{(i)}\}_{i \in \lset_{X\mid Y}^{\rm c}}$, we can recover the whole vector $U^{1:n}$ in a successive manner.

To construct a polar code for the channel $W$, we proceed as
follows. The information is placed in the positions indexed by $\iset
= \hset_X \cap \lset_{X\mid Y}$. By using \eqref{eq:cardint1}, \eqref{eq:cardint2} and that $\hset_{X \mid Y} \subseteq \hset_X$, it follows that
\begin{equation}\label{eq:I}
\lim_{n\to\infty}\frac{1}{n} \, | \iset| = H(X)-H(X \mid Y) = I(X;Y).
\end{equation}
Hence, our requirement on the transmission rate is met.

The remaining positions are frozen. More precisely, they are divided
into two subsets, namely $\fset_{\rm r} = \hset_X \cap \lset_{X \mid
Y}^{\text{c}}$ and $\fset_{\rm d} = \hset_X^{\text{c}}$.  For $i\in \fset_{\rm r}$, $U^{(i)}$ is
independent of $U^{1:i-1}$, but it cannot be reliably decoded using
$Y^{1:n}$. Hence, these positions are filled with bits chosen uniformly
at random, and this randomness is assumed to be shared between the
transmitter and the receiver. For $i \in \fset_{\rm d}$,
the value of $U^{(i)}$ has to be chosen in a particular way. Indeed, almost all these positions are in $\lset_X$, hence $U^{(i)}$ is approximately
a deterministic function of $U^{1:i-1}$. Below, we discuss in detail how to choose the values associated with the positions in $\fset_{\rm d}$. The situation is schematically
represented in Figure~\ref{fig:assymetric}.

\begin{figure}[tb]
\begin{center} \includegraphics[width=6.5cm]{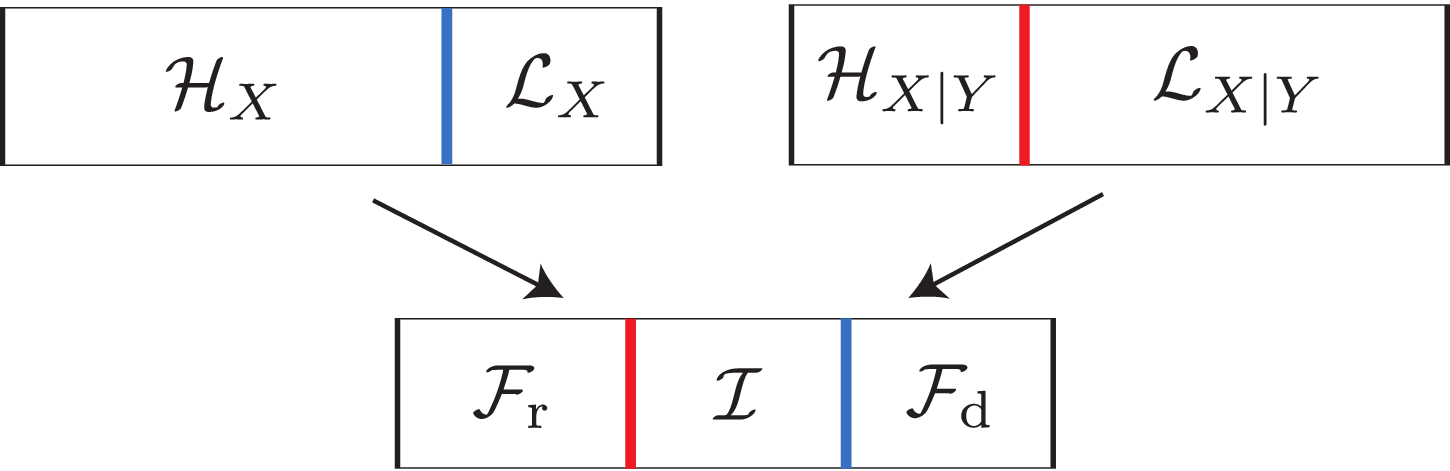}
\end{center}
\caption{Coding over asymmetric channels via an integrated scheme based on polar codes. The top left rectangle represents the subdivision of indices that yields the source coding part of the scheme. The top right rectangle represents the subdivision of indices that yields the channel coding part of the scheme. As a result, the set of
indices $[n]$ can be partitioned into three subsets (bottom image):
the information indices $\mathcal{I} = \hset_X \cap \lset_{X\mid Y}$; the
frozen indices $\mathcal{F}_{\rm r} = \hset_X \cap \lset_{X \mid Y}^{\text{c}}$
filled with uniformly random bits; and the frozen indices
$\mathcal{F}_{\rm d} = \hset_X^{\text{c}}$ chosen according to either a ``randomized rounding'' rule or an ``argmax'' rule.}
\label{fig:assymetric} \end{figure}

\vspace{1em}

\noindent {\bf Encoding.} We place the information into the positions indexed by $\iset$, hence let $\{u^{(i)}\}_{i \in \iset}$ denote the bits to be transmitted. Then, we fill the positions indexed by $\fset_{\rm r}$ with a random sequence that is shared between the transmitter
and the receiver, hence let $\{u^{(i)}\}_{i \in \fset_{\rm r}}$ be the particular realization of this sequence. 

Let us now consider the encoding of the positions in $\fset_{\rm d}$. An analogous problem was first addressed in \cite[Section III]{KoU09}, where polar codes were used for lossy source coding. The same problem also occurs when constructing polar codes for the broadcast channel \cite{GAG13ar, MHSU14}. There are at least two possible approaches for solving this issue.

On the one hand, we can use a ``randomized rounding'' rule, that consists of setting the value of the $i$-th bit according to the distribution ${\mathbb P}_{U^{(i)} \mid U^{1:i-1}}$. In formulae, 
\begin{equation}\label{eq:randround}
u^{(i)} = \left\{ \begin{array}{ll} 0, & \mbox{ w.p. } {\mathbb P}_{U^{(i)} \mid U^{1:i-1}}(0 \mid  u^{1:i-1}) \\ 1, & \mbox{ w.p. } {\mathbb P}_{U^{(i)} \mid U^{1:i-1}}(1 \mid  u^{1:i-1})\end{array} \right.
\end{equation}
The random number generator used to construct this sequence is shared between the transmitter and the receiver. The ``randomized rounding'' rule yields a provable result (see \cite{HY13} for asymmetric channels and \cite{GAG13ar} for broadcast channels). 

On the other hand, we can use an ``argmax'' rule, that consists of setting the $i$-th bit to the value that maximizes ${\mathbb P}_{U^{(i)} \mid U^{1:i-1}}$. In formulae, 
\begin{equation}
u^{(i)} = \arg\max_{u\in \{0, 1\}} {\mathbb P}_{U^{(i)} \mid U^{1:i-1}}(u \mid  u^{1:i-1}).
\end{equation}
The ``argmax'' rule seems to perform slightly better in numerical simulations, and in the recent work \cite{CB15} some progress has been made concerning its theoretical analysis. In particular, it has been proved that setting the low-entropy indices, i.e., the indices in $\lset_X$, according to the ``argmax'' rule, while still keeping a ``randomized rounding'' rule for the indices in $\hset_X^{\rm c}\setminus \lset_X$, yields a provable result. Note further that $|\hset_X^{\rm c}\setminus \lset_X|=o(n)$, i.e., the size of this last set is sublinear in $n$.

Eventually, the elements of $\{u^{(i)}\}_{i \in \fset_{\rm d}}$ are computed in successive order either using a ``randomized rounding'' or an ``argmax'' rule, and the probabilities ${\mathbb P}_{U^{(i)} \mid U^{1:i-1}}(u \mid  u^{1:i-1})$ can be obtained recursively with complexity $\Theta(n \log n)$. Since $G_n = G_n^{(-1)}$, the vector $(x^{1:n})^T = (u^{1:n})^T G_n$ is transmitted
over the channel.

\vspace{1em}

\noindent {\bf Decoding.} The decoder receives $y^{1:n}$ and computes the
estimate $\hat{u}^{1:n}$ of $u^{1:n}$ according to the rule
\begin{equation}\label{eq:decruleas}
\hat{u}^{(i)} = \left\{ \begin{array}{ll}
u^{(i)}, & \mbox{if } i \in \fset_{\rm r} \\
\displaystyle\arg\max_{u\in \{0, 1\}}
{\mathbb P}_{U^{(i)} \mid U^{1:i-1}}(u \mid  u^{1:i-1}),&
\mbox{if } i \in \fset_{\rm d} \\
\displaystyle\arg\max_{u\in \{0, 1\}} {\mathbb P}_{U^{(i)}
\mid U^{1:i-1}, Y^{1:n}}(u \mid u^{1:i-1}, y^{1:n}), & 
\mbox{if } i \in \iset
\\ \end{array}\right.  \end{equation} where ${\mathbb P}_{U^{(i)} \mid
U^{1:i-1}, Y^{1:n}}(u \mid u^{1:i-1}, y^{1:n})$ can be computed recursively with complexity $\Theta(n \log n)$. In \eqref{eq:decruleas}, we assume that the ``argmax'' rule is used to encode the positions in $\fset_{\rm d}$. If the ``randomized rounding'' rule is used, then the decoder can still correctly recover $u^{(i)}$, since the random sequence used in \eqref{eq:randround} is shared between the transmitter and the receiver.  

\vspace{1em}

\noindent {\bf Performance.} The block error probability $P_{\rm e}$ can be upper bounded by
\begin{equation}\label{eq:perfop}
P_{\rm e} \stackrel{\mathclap{\mbox{\footnotesize(a)}}}{\le} \sum_{i \in \iset} Z(U^{(i)} \mid U^{1:i-1}, Y^{1:n}) \stackrel{\mathclap{\mbox{\footnotesize(b)}}}{=} O(2^{-n^{\beta}}),
\quad \forall \, \beta \in (0, 1/2).
\end{equation}
Let us briefly comment on how to obtain formula \eqref{eq:perfop}. The inequality (a) comes from the union bound: the error probability under SC decoding is upper bounded by the sum of the probabilities of making a mistake while decoding each of the information bits (the remaining bits are frozen, therefore they are known at the decoder). Furthermore, the probability of decoding incorrectly the $i$-th synthetic channel ($i\in \mathcal I$) is upper bounded by its Bhattacharyya parameter, namely, $Z(U^{(i)} \mid U^{1:i-1}, Y^{1:n})$. For a full proof of this formula, see also the proof of Theorem 3 in \cite{HY13}. Finally, the equality (b) comes from the definition of the set $\mathcal I$, that contains positions $i$ s.t. $Z(U^{(i)} \mid U^{1:i-1}, Y^{1:n})$ is small enough.

\subsection{Sparse Graph Codes}\label{subsec:sparse}

\noindent {\bf Design of the Scheme.} Consider a linear code with parity-check
matrix $P$ with $nH(X)=nh_2(\alpha)$ rows and $n$ columns, namely
$P \in {\mathbb F}_2^{nh_2(\alpha)\times n}$. Let $X^{1:n} \in {\mathbb F}_2^{n}$
be a codeword and denote by $Y^{1:n}$ the corresponding channel output. Let $S^{1:nh_2(\alpha)}
\in {\mathbb F}_2^{nh_2(\alpha)}$ be the vector of syndromes defined as $S^{1:nh_2(\alpha)} =
P X^{1:n}$.

Recall that, in the integrated scheme, we need to achieve the source
coding and the channel coding part at the same time. To do so, we divide $S^{1:nh_2(\alpha)}$ into two parts, i.e.,
\begin{equation}\label{eq:decS}
S^{1:nh_2(\alpha)}
= (S_{1}^{1:nC(W)}, S_2^{1: nH(X\mid Y)})^T,
\end{equation}
where this decomposition is possible because $h_2(\alpha) = H(X) = C(W)+H(X\mid Y)$. Similarly, it is convenient to write the parity-check matrix $P$ as
\begin{equation}\label{eq:decP}
P= [P_1, P_2]^T,
\end{equation}
where $P_1\in {\mathbb F}_2^{nC(W)\times n}$ and $P_2\in {\mathbb F}_2^{nH(X\mid Y)\times n}$.

The first part of the decomposition \eqref{eq:decS}, namely
$S_{1}^{1:nC(W)}$, contains the information
bits. This is quite different from what happens in a standard parity-check code, in which the values of the parity checks are shared between the encoder and the decoder (and typically fixed to 0). In the
proposed scheme, the parity checks contain the transmitted message.

The second part, namely $S_2^{1: nH(X\mid Y)}$,
is chosen uniformly at random, and this randomness is
assumed to be shared between the transmitter and the receiver. Note that $S_2^{1: nH(X\mid Y)}$ does not depend on the information bits.

The choice of the parity-check matrix $P_2$ concerns the \emph{channel
coding} part of the scheme. Recall the problem considered in Section
\ref{subsec:checks}: given the channel output $Y^{1:n}$ and the parity bits $S_2^{1: nH(X\mid Y)}$, the receiver can reconstruct $X^{1:n}$, as long as the parity-check matrix corresponds to a code that achieves the capacity of the ``symmetrized
channel'' with density given by \eqref{eq:ldenssymmpost}. For example, we can set $P_2$ to be the parity-check matrix of an
$(\dl, \dr)$-regular SC-LDPC ensemble with sufficiently large
degrees.

The choice of the parity-check matrix $P_1$ concerns the \emph{source
coding} part of the scheme. In particular, we choose $P_1$ in order to fulfill the following requirement: we want to associate with each syndrome $S^{1:nh_2(\alpha)}$ a codeword $X^{1:n}$ with $[P_1, P_2]^T X^{1:n} = S^{1:nh_2(\alpha)}$ so that the uniform i.i.d. distribution on the syndromes induces a Bernoulli$(\alpha)$ i.i.d. distribution on the codewords. 

Before moving on with the description of the scheme, let us review how to use sparse graph codes to accomplish lossless source coding. We are given a vector $X^{1:n}$ of $n$ i.i.d. Bernoulli($\alpha$)
random variables and the aim is to compress it into a binary sequence of size roughly $nh_2(\alpha)$.
We want to solve the problem by using the parity-check matrix $\tilde{P}\in {\mathbb F}_2^{nh_2(\alpha)\times n}$
of a sparse graph code as the linear compressor and the BP decoder as the
decompressor, respectively \cite{CSV04dimacs}. More specifically, given
$x^{1:n}$ to be compressed, the encoder computes
$s^{1:nh_2(\alpha)} = \tilde{P} x^{1:n}$. The task of the decoder
can be summarized as follows:

\vspace{0.15em}

\emph{Task 1.} Given the syndrome vector $\tilde{P} x^{1:n}$, recover the biased vector $x^{1:n}$ by using the BP algorithm.

\vspace{0.15em}

Let us now relate this task to a channel coding problem. Let $c^{1:n}$ be a codeword of the code with parity-check matrix $\tilde{P}$, i.e., $\tilde{P} c^{1:n}= 0^{1:nh_2(\alpha)}$. Consider the transmission of $c^{1:n}$ over the binary symmetric channel with crossover probability $\alpha$, i.e., the BSC$(\alpha)$, and let $y^{1:n}$ be the channel output. Denote by $\tilde{P} y^{1:n}$ the syndrome computed by the decoder, and note that $\tilde{P} y^{1:n} = \tilde{P}e^{1:n}$, where $e^{(i)}=1$ if the $i$-th bit was flipped by the channel, and $0$ otherwise. Consider the following two tasks:

\vspace{0.15em}
\emph{Task 2.} Given the syndrome vector $\tilde{P} e^{1:n}$, recover the error vector $e^{1:n}$ by using the BP algorithm.
	
\vspace{0.15em}
\emph{Task 3.} Given the received vector $y^{1:n}$, recover the transmitted codeword $c^{1:n}$ by using the BP algorithm.
\vspace{0.15em}

Let us briefly show that these three tasks are, in fact, equivalent. First of all, note that \emph{Task 1} and \emph{Task 2} are clearly identical. Furthermore, it is shown in \cite{CSV04dimacs}  that \emph{Task 2} succeeds if and only if \emph{Task 3} succeeds. The idea is to write down the message-passing equations in the two cases, and to observe that the messages obtained in \emph{Task 2} can be put in one-to-one correspondence with the messages obtained in \emph{Task 3}. More specifically, on the one hand, in \emph{Task 2} we initialize all the received values at variable nodes by $\ln(\bar{\alpha}/\alpha)$ and the check nodes have an associated sign given by the vector $(-1)^{\tilde{P} e^{1:n}}$. On the other hand, in \emph{Task 3}, we initialize the received values at variable nodes by the vector $(-1)^{y^{1:n}} \cdot \ln(\bar{\alpha}/\alpha)$, and all the check nodes have an associated sign of $+1$. The crucial observation is that, for each iteration of the BP algorithm, the modulus of the received values at variable nodes stays the same for the two tasks, and the sign is flipped according to the value of $y^{1:n}$.

Note that \emph{Task 3} is the standard channel coding problem for the transmission over the BSC. Hence, we can use the parity-check matrix of a code that achieves capacity over the BSC to compress $n$ i.i.d. Bernoulli($\alpha$) random variables into a binary sequence of size roughly $nh_2(\alpha)$. 

Let us come back to our original problem of achieving the capacity
of a B-DMC. The source coding part of our approach is basically the
inverse of source coding. Indeed, given the uniform vector of syndromes $S^{1:nh_2(\alpha)}$, we want to obtain a biased codeword $X^{1:n}$.

Let $P_1$ be the parity-check matrix of a regular SC-LDPC ensemble with sufficiently large degrees. This implies that also $P= [P_1, P_2]^T$ is the parity-check matrix of a regular SC-LDPC ensemble\footnote{Note that, even though $P$ is not chosen according to the usual definition of an SC-LDPC ensemble, it still has the same asymptotic performance.}. First, suppose that the vector of syndromes to be fulfilled has size $m$ slightly larger than $nh_2(\alpha)$, say $m=n(h_2(\alpha)+\epsilon)$ for some small $\epsilon > 0$. Consequently, suppose that the matrix $P$ has $m$ rows. Clearly, given a vector  $X^{1:n}$ of $n$ i.i.d. Bernoulli($\alpha$) random variables, we can always define $S^{1:m}$ as $P X^{1:n} = S^{1:m}$. However, only for a vanishing fraction of possible vectors of syndromes $S^{1:m}$ there exists $X^{1:n}$ s.t. $P X^{1:n} = S^{1:m}$. This means that, for a randomly chosen $S^{1:m}$, with high probability the BP algorithm will not succeed. 

Suppose now that the vector of syndromes to be fulfilled has size $m$ no larger than $nh_2(\alpha)$. Then, with high probability, there are exponentially many $X^{1:n}$ with i.i.d. Bernoulli$(\alpha)$ distribution s.t. $P X^{1:n}= S^{1:m}$. This implies that the BP algorithm does not converge, as a message-passing decoder operating locally can easily get confused when there are many feasible solutions. 

Perhaps a more apt approach is to frame the source coding part of our scheme as a lossy compression problem, where the distortion between the distribution of $X^{1:n}$ and an i.i.d. Bernoulli$(\alpha)$ distribution tends to $0$ as $n$ goes large. It was observed in \cite{CM05} that using a standard BP algorithm is not effective for lossy compression, and that this issue can be overcome by introducing a decimation process. An encoding scheme for lossy compression based on spatially coupled low-density generator-matrix (LDGM) codes and belief-propagation guided decimation is presented in \cite{AMV15}, where it is shown with numerical simulations that the spatially coupled ensemble approaches the Shannon rate-distortion limit for large check degrees. This technique is extended to the Wyner-Ziv and Gelfand-Pinsker problems in \cite{KVNP14}, where it is shown empirically that spatially coupled compound LDGM/LDPC codes with belief-propagation guided decimation achieve the optimal rates. In particular, the solution to the Gelfand-Pinsker problem presents some similarities to our approach: the information bits are placed in a vector of syndromes, and the compound LDGM/LDPC codes are simultaneously good for rate distortion and channel coding. The need for a scheme that is good both for source and channel coding is due to the fact that, in the Gelfand-Pinsker setting, there is a constraint on the average weight of the transmitted codeword. This is analogous to our requirement that $X^{1:n}$ has a Bernoulli$(\alpha)$ i.i.d. distribution.

As mentioned earlier, other schemes adopting a framework similar to the one considered in this section are provided in \cite{MuMi08, MuMi09, MuMi10}. In particular, in \cite{MuMi10} the authors consider LDPC matrices with logarithmic column weight and maximum likelihood decoding. This approach provably achieves the optimal rate by introducing the notion of a hash property (but the decoding algorithm has exponential complexity). The results of \cite{MuMi10} are extended in \cite{Mu14}, where codes for general (thus, possibly asymmetric) channels and sources are constructed. It is interesting to point out that the problem of generating the codeword $X^{1:n}$ is solved in \cite{Mu14} with a constrained-random-number generator, instead of resorting to a belief-propagation type of algorithm.

Our solution follows the lead of \cite{AMV15, KVNP14} and uses belief-propagation guided decimation at the encoder. Note that this approach works well in practice, as testified by the simulation results in \cite{AMV15, KVNP14}, but we currently have no theoretical guarantees on its performance. Let us now get down to the details of the proposed encoding scheme. Given the syndrome vector $S^{1:nh_2(\alpha)}$, we run the standard BP algorithm and, after every $t$ iterations, for some fixed $t \in {\mathbb N}$, we decimate a small fraction of the codeword bits. This means that we set each decimated bit to its most likely value. Furthermore, we fix the modulus of the received values at the corresponding variable nodes to $+\infty$. Consequently, the decimated bits will not change during the next iterations of the algorithm. The procedure ends when all the codeword bits have been decimated.

Let us denote by $\tilde{X}^{1:n}$ the codeword produced as output by the algorithm described above. Note that, because of the decimation steps, it is possible that $[P_1 P_2]^T \tilde{X}^{1:n}$ differs from $S^{1:nh_2(\alpha)}$ in some positions. Furthermore, recall that $S_1^{1:nC(W)}$ contains the information bits. Hence, if $P_1 \tilde{X}^{1:n}$ differs from $S_1^{1:nC(W)}$, even if the decoder correctly reconstructs the transmitted codeword $\tilde{X}^{1:n}$, it will not correctly reconstruct the information sequence. However, the idea is that the fraction of positions in which $P_1 \tilde{X}^{1:n}$ differs from $S_1^{1:nC(W)}$ tends to $0$ as $n$ goes large. This intuition is confirmed by the numerical simulations in \cite{AMV15, KVNP14}. Hence, in order to cope with this issue, we pre-code $S_1^{1:nC(W)}$ with a negligible loss in rate. Note that the specific choice of this pre-code does not affect much the overall performance of the proposed scheme since, as $n$ goes large, the rate of the pre-code is expected to tend to $1$.

\vspace{1em}

\noindent {\bf Encoding.} First, we pre-code the $nC(W)$ information bits with a pre-code $\mathcal C_{\rm p}$ of rate close to $1$. We can use, for instance, an SC-LDPC code or a polar code designed for transmission over the BSC. The output of this pre-coding operation is the sequence $s_{1}^{1:n(C(W)+\epsilon)}$, for some small $\epsilon>0$. Then, we fill $s_2^{1: nH(X\mid Y)}$ with a realization of a sequence chosen uniformly at random and shared between the transmitter and the receiver. Let $P$ be the parity-check matrix of an SC-LDPC code with sufficiently large degrees. From the syndrome vector $s^{1:n(H(X)+\epsilon)} = (s_{1}^{1:n(C(W)+\epsilon)}, s_2^{1: nH(X\mid Y)})^T$ and the parity-check matrix $P$, we obtain the codeword $\tilde{x}^{1:n}$ by running the BP algorithm with decimation steps.

Let $P_1$ be obtained by decomposing $P$ as in \eqref{eq:decS}. Let us now check that we can recover correctly the initial information bits from $P_1 \tilde{x}^{1:n}$ by decoding $\mathcal C_{\rm p}$. If this is not possible, then the overall procedure is repeated with a different choice for the code $\mathcal C_{\rm p}$. Once the decoding of $\mathcal C_{\rm p}$ succeeds, we transmit the vector $\tilde{x}^{1:n}$ over the channel.

The final choice of the pre-code $\mathcal C_{\rm p}$ is shared between the transmitter and the receiver, e.g., by common randomness. Note that the pre-code $\mathcal C_{\rm p}$ needs to be shared once for all before the actual transmission starts.

\vspace{1em}

\noindent {\bf Decoding.} The decoder receives $y^{1:n}$ and runs the BP algorithm using also the vector of syndromes $s_2^{1: nH(X\mid Y)}$ shared with encoder. Let $\hat{x}^{1:n}$ be the output of the BP algorithm. Eventually, we recover the information bits from $P_1 \hat{x}^{1:n}$ by decoding $\mathcal C_{\rm p}$. 

\vspace{1em}

\noindent {\bf Performance.} There are two possible types of errors. First, the encoder might fail to produce a codeword $\tilde{x}^{1:n}$ for which the decoding of $\mathcal C_{\rm p}$ succeeds. Second, the decoder might not estimate correctly the transmitted vector, i.e., $\hat{x}^{1:n} \neq \tilde{x}^{1:n}$. Note that, by construction of $\tilde{x}^{1:n}$, if $\hat{x}^{1:n} = \tilde{x}^{1:n}$, then the decoder recovers correctly the information bits. 

The second error event occurs with vanishing probability, and this result is provable by following the argument of Section \ref{subsec:checks}. Concerning the first error event, we only need that it does not happen with probability $1$, because the encoding operation can be attempted several times. As we previously pointed out, from numerical simulations we observe that $P_1 \tilde{x}^{1:n}$ agrees with $s_{1}^{1:n(C(W)+\epsilon)}$ in almost all the positions. We remark that, to the best of our knowledge, this last statement is not provable, because of the decimation steps introduced in the BP algorithm. If $P_1 \tilde{x}^{1:n}$ and $s_{1}^{1:n(C(W)+\epsilon)}$ are sufficiently close, then, with high probability, we can recover the information bits by decoding $\mathcal C_{\rm p}$.

\section{Paradigm 3: Chaining Construction}\label{sec:chaining}

In the integrated approach, discussed in the preceding section,
the idea was to use a single code to accomplish both the source and
the channel coding part. The chaining construction, on the contrary, enables us to separate
these two tasks. In particular, we provide conditions under which a source code and a channel code can be combined in order to achieve the capacity of any DMC. This idea was first presented in
\cite{BoM11}. Here, we prove that it can be used to achieve the 
capacity of any asymmetric channel. 

The problem statement is the same as in Section \ref{sec:integrated}. Our main idea is formalized by the following theorem. 

\begin{theorem}\label{th:chain}
Let $W$ be a B-DMC with capacity-achieving input distribution $\{p^*(x)\}_{x\in \{0, 1\}}$ s.t. $p^*(1)=\alpha$ for some $\alpha \in [0, 1]$. Denote by $X$ and $Y$ the input and the output of the channel, respectively. Let $n, m, \ell \in \mathbb N$, where $n$ is sufficiently large, $m$ is roughly $nh_2(\alpha)$, and $\ell$ is roughly $nH(X\mid Y)$. Denote by $U^{1:m}$ a sequence of $m$ i.i.d. uniform random variables, and by $X^{1:n}$ a sequence of $n$ i.i.d. Bernoulli$(\alpha)$ random variables.
Let $Y^{1:n}$ be the channel output when $X^{1:n}$ is transmitted. 
Assume that, for any $\delta>0$, there exists $n \in \mathbb N$ and there exist maps $f: \{0, 1\}^n \rightarrow \{0, 1\}^m$, $g: \{0, 1\}^m \rightarrow \{0, 1\}^n$, and $h: \{0, 1\}^n \rightarrow \{0, 1\}^\ell$ that satisfy the following properties.
\begin{enumerate}
\item $U^{1:m}=f(g(U^{1:m}))$, i.e., the map $f$ is invertible, with probability $1-\delta$.
\item The total variation distance between the distribution of $g(U^{1:m})$ and the distribution of $X^{1:n}$ is upper bounded by $\delta$.
\item Given $Y^{1:n}$ and $h(X^{1:n})$, it is possible to reconstruct $X^{1:n}$ with probability $1-\delta$.
\end{enumerate}
Then, we can use $f$, $g$, and $h$ to transmit over $W$ with rate close to $C(W)$.
\end{theorem}

In the following, we will prove this theorem and we will provide choices of $f$, $g$, and $h$ that fulfill the required properties. 

\vspace{1em}

\noindent {\bf Design of the Scheme.} First, we consider the {\em source coding}
part of the scheme. Recall that in the previous section we framed this task as the inverse of source coding and we described a solution that uses sparse graph codes and belief-propagation guided decimation. Let us now consider this problem from a more general point of view.

In the traditional lossless compression setting, the input is a
sequence $X^{1:n}$ with i.i.d. Bernoulli$(\alpha)$ distribution, and the encoder consists of a map from the set $\{0, 1\}^n$ of source sequences to the set $\{0,1\}^*$ of finite-length binary
strings. Let $f:\{0, 1\}^n\to \{0, 1\}^*$ be the encoding map,
so that $U = f(X^{1:n})$ is the compressed description of the source
sequence $X^{1:n}$. For a good source code, the expected binary length
of $U$ is close to the entropy of the source, i.e., $nh_2(\alpha)$.  In addition, the decoding
function $g:\{0, 1\}^*\to \{0, 1\}^n$ is such that $X^{1:n}=g(f(X^{1:n}))$ with high probability over the choice of $X^{1:n}$. Several solutions to this problem
have been proposed to date, such as, Huffman coding, arithmetic coding,
Lempel-Ziv compression \cite{CoT06}, polar codes \cite{Ar10, CrKo10}, and
LDPC codes \cite{CSV04dimacs}, just to name a few.

In our setting, the input is the compressed sequence $U^{1:m}$ with i.i.d. uniform distribution, instead of the biased sequence $X^{1:n}$. Note that $U^{1:m}$ contains the information bits. Furthermore, we consider maps from $\{0, 1\}^n$ to $\{0, 1\}^m$ and vice versa, where $m$ is a fixed integer of size roughly $n h_2(\alpha)$, instead of maps from $\{0, 1\}^n$ to $\{0, 1\}^*$. More specifically, the encoder and the decoder implement the maps $g: \{0, 1\}^m \rightarrow \{0, 1\}^n$ and $f: \{0, 1\}^n \rightarrow \{0, 1\}^m$, respectively. This problem can also be regarded as an instance of \emph{homophonic coding}. Indeed, homophonic coding is a framework to convert a sequence with some probability distribution into an invertible sequence with a different probability distribution. Such a framework was first considered in the context of cryptography \cite{Gu88, JKM89}, in order to ensure that any sequence of ciphertext appears with the same frequency. More specifically, in a few paragraphs we will consider the interval algorithm for homophonic coding scheme presented in \cite{HoHa01}. Furthermore, homophonic coding has been used to generate biased codewords in the context of LDPC codes \cite[Chapter 4]{Hon13} and to construct polar codes for channels with memory \cite{HoYa15}. 

For our scheme, we require that the maps $f$ and $g$ satisfy the first two properties stated in Theorem \ref{th:chain}. Let us justify such requirements.

The first property ensures that, given $g(U^{1:m})$, it is possible to recover $U^{1:m}$ with high probability. This requirement is crucial because $g(U^{1:m})$ represents the codeword transmitted over the channel. Hence, at the decoder, given the channel output, we estimate $g(U^{1:m})$ and, from this, we deduce the information vector $U^{1:m}$. 

The second property ensures that the error probability for the transmission of $g(U^{1:m})$ is upper bounded by the error probability for the transmission of $n$ i.i.d. Bernoulli$(\alpha)$ random variables plus $\delta$. This statement can be proved as follows. Recall that, by definition of total variation distance, the second property can be written as
\begin{equation}\label{eq:TVdist2}
\frac{1}{2}\sum_{x\in \{0, 1\}^n}|{\mathbb P}_{g(U^{1:m})}(x)-{\mathbb P}_{X^{1:n}}(x)| < \delta,
\end{equation}
where $X^{1:n}$ has an i.i.d. Bernoulli$(\alpha)$ distribution. Then, by using that
\begin{equation*}
\sum_{x\in \{0, 1\}^n}{\mathbb P}_{g(U^{1:m})}(x) = \sum_{x\in \{0, 1\}^n}{\mathbb P}_{X^{1:n}}(x) = 1,
\end{equation*}
we obtain
\begin{equation}\label{eq:TVdist3}
\sum_{x\in \{0, 1\}^n}\max({\mathbb P}_{g(U^{1:m})}(x)-{\mathbb P}_{X^{1:n}}(x), 0) < \delta.
\end{equation}
Denote by $P_{\rm e}$ and $\tilde{P}_{\rm e}$ the error probabilities when the transmitted codeword is distributed according to $g(U^{1:m})$ and $X^{1:n}$, respectively. Then,
\begin{equation}\label{eq:req2}
\begin{split}
P_{\rm e} &= \sum_{x\in \{0, 1\}^n} {\mathbb P}(\mbox{error} \mid x){\mathbb P}_{g(U^{1:m})}(x) = \sum_{x\in \{0, 1\}^n} {\mathbb P}(\mbox{error} \mid x)({\mathbb P}_{g(U^{1:m})}(x)-{\mathbb P}_{X^{1:n}}(x))+ \sum_{x\in \{0, 1\}^n} {\mathbb P}(\mbox{error} \mid x){\mathbb P}_{X^{1:n}}(x)\\
&\le \sum_{x\in \{0, 1\}^n} {\mathbb P}(\mbox{error} \mid x)\cdot \max({\mathbb P}_{g(U^{1:m})}(x)-{\mathbb P}_{X^{1:n}}(x), 0)+ \sum_{x\in \{0, 1\}^n} {\mathbb P}(\mbox{error} \mid x){\mathbb P}_{X^{1:n}}(x)\\
&<\delta + \tilde{P}_{\rm e},
\end{split}
\end{equation}
where the last inequality uses \eqref{eq:TVdist3} and that ${\mathbb P}(\mbox{error} \mid x)\le 1 $. 
The requirement \eqref{eq:req2} is crucial because, in the channel coding part of the scheme, we assume that the transmitted codeword has an i.i.d. Bernoulli$(\alpha)$ distribution.

Let us now describe how to construct maps $f$ and $g$ such that these maps satisfy the desired properties. One possible solution is based on polar codes, and the idea follows closely the scheme described in Section \ref{subsec:polar}. Given $U^{1:m}$ as input, we put it into the positions indexed by $\hset_X$ defined in \eqref{lossless_sets}, and set the remaining positions according to the ``randomized rounding'' rule \eqref{eq:randround}. Then, we multiply this vector with the matrix $G_n$, and define $g(U^{1:m})$ as the result of this last operation. Given $X^{1:n}$ as input, we multiply it with the matrix $G_n$, and extract the positions indexed by $\hset_X$. We define $f(X^{1:n})$ as the result of this last operation. It is clear that $U^{1:m}=f(g(U^{1:m}))$, hence the first property of Theorem \ref{th:chain} is satisfied. By following the proof of Theorem 3 of \cite{HY13}, we also obtain that the total variation distance between the distribution of $g(U^{1:m})$ and the i.i.d. Bernoulli$(\alpha)$ distribution is upper bounded by $2^{-n^{\beta}}$ for $\beta<1/2$. Hence,  the second property of Theorem \ref{th:chain} is satisfied.

An alternative solution coincides with a special case of the interval algorithm for fixed-to-fixed homophonic coding proposed in Section III-B of \cite{HoHa01}. Let us start by defining the map $g$. We partition the interval $[0, 1)$ into $2^m$ sub-intervals of size $2^{-m}$. Given $U^{1:m}$ as input, we interpret this sequence as the integer $K\in \{0, \cdots, 2^m-1\}$, and we pick a point $\bar{U}$ uniformly at random in the sub-interval $[K2^{-m}, (K+1)2^{-m})$. The output sequence $g(U^{1:m})$ is obtained from $\bar{U}$ as follows. Given an interval $I=[i_{\rm start}, i_{\rm end})$, we partition it into the sub-intervals $I_1=[i_{\rm start}, i_{\rm start} + \alpha(i_{\rm end}-i_{\rm start}))$ and $I_2 = [i_{\rm start} + \alpha(i_{\rm end}-i_{\rm start}), i_{\rm end})$. Note that $|I_1| = \alpha |I|$ and $|I_2| = (1-\alpha)|I|$, where $|I|$, $|I_1|$, and $|I_2|$ denote the sizes of $I$, $I_1$, and $I_2$, respectively. We initialize $I$ to be the interval $[0, 1)$. If $\bar{U} \in I_1$, then we output a $1$ and redefine $I$ to be $I_1$; otherwise, we output a $0$ and redefine $I$ to be $I_2$. By repeating this procedure $n$ times, we obtain the sequence $g(U^{1:m})$.

Let us define the map $f$. Given $X^{1:n}$ as input, we evaluate the interval $I$ according to the following iterative procedure. We initialize $I$ to be the interval $[0, 1)$. If the input is $1$, we redefine $I$ to be $I_1$; otherwise, we redefine $I$ to be $I_2$. As the input sequence $X^{1:n}$ has length $n$, we repeat this operation $n$ times. Then, we pick a point $\bar{X}$ uniformly at random in the resulting interval $I$. Let $K$ be such that $\bar{X} \in [K2^{-m}, (K+1)2^{-m})$. We define $f(X^{1:n})$ to be the sequence associated to the integer $K$.

Let us prove that the maps $f$ and $g$ defined above satisfy the desired properties. As $U^{1:m}$ is a sequence of $m$ i.i.d. uniform random variables, the point $\bar{U}$ is uniformly distributed in $[0, 1)$. Hence, the sequence $g(U^{1:m})$ obtained with the aforementioned procedure is exactly a sequence of $n$ i.i.d. Bernoulli$(\alpha)$ random variables. As a result, the second property of Theorem \ref{th:chain} holds.

Let us prove that also the first property holds. Recall that to each sequence $u\in \{0, 1\}^m$ is associated an interval $I(u)$ of size $|I(u)|=2^{-m}$. Since all these intervals have the same size, we say that they are \emph{even}. Furthermore, to each sequence $x\in \{0, 1\}^n$ is associated an interval $I(x)$ of size $|I(x)|=\alpha^{w_{\rm H}(x)}(1-\alpha)^{n-w_{\rm H}(x)}$, where $w_{\rm H}(\cdot)$ denotes the Hamming weight. Since these intervals do not have the same size, we say that they are \emph{odd}. The crucial observation is the following: if the random variable $\bar{U}$ falls into an odd interval that is entirely contained into an even interval, then $f(g(U^{1:m}))=U^{1:m}$. Denote by $\mathcal I_{\rm bad}$ the set of \emph{bad odd} intervals, i.e., the set of odd intervals that are not entirely contained into an even interval. Then, for any $\epsilon > 0$,
\begin{equation}\label{eq:arithmetic}
\begin{split}
\mathbb{P}(f(g(U^{1:m}))\neq U^{1:m}) &\le \mathbb{P}(\bar{U}\in \mathcal I_{\rm bad}) \\
&\stackrel{\mathclap{\mbox{\footnotesize(a)}}}{=} \sum_{x\in \{0, 1\}^n : I(x)\in I_{\rm bad}} |I(x)| \\
&= \sum_{\substack{x\in \{0, 1\}^n : I(x)\in I_{\rm bad}\\ |I(x)| > 2^{-n(h_2(\alpha)-\epsilon)}}} |I(x)| + \sum_{\substack{x\in \{0, 1\}^n : I(x)\in I_{\rm bad}\\ |I(x)| \le 2^{-n(h_2(\alpha)-\epsilon)}}} |I(x)|, 
\end{split}
\end{equation} 
where in (a) we use that $\bar{U}$ is uniformly distributed in $[0, 1)$. The first term of the RHS of \eqref{eq:arithmetic} can be upper bounded as
\begin{equation}\label{eq:sum1}
\begin{split}
\sum_{\substack{x\in \{0, 1\}^n : I(x)\in I_{\rm bad}\\ |I(x)| > 2^{-n(h_2(\alpha)-\epsilon)}}} |I(x)| &\le \sum_{\substack{x\in \{0, 1\}^n \\ |I(x)| > 2^{-n(h_2(\alpha)-\epsilon)}}} |I(x)| \\
&\stackrel{\mathclap{\mbox{\footnotesize(a)}}}{=} \sum_{\substack{x\in \{0, 1\}^n \\ |I(x)| > 2^{-n(h_2(\alpha)-\epsilon)}}}\mathbb{P}(X^{1:n}=x) \\
&= \mathbb{P}\left(\alpha^{w_{\rm H}(X^{1:n})}(1-\alpha)^{n-w_{\rm H}(X^{1:n})}> 2^{-n(h_2(\alpha)-\epsilon)}\right),
\end{split}
\end{equation}
where in (a) we use that $\mathbb{P}(X^{1:n}=x)=|I(x)|$. As $X^{1:n}$ has an i.i.d. Bernoulli$(\alpha)$ distribution, we have that, for any $\epsilon_1 >0$,
\begin{equation*}
\mathbb{P}\left(w_{\rm H}(X^{1:n})\not\in(nh_2(\alpha)-\epsilon_1, nh_2(\alpha)+\epsilon_1)\right)\underset{n\to\infty}{\longrightarrow} 0.
\end{equation*}
Hence, the RHS of \eqref{eq:sum1} tends to $0$, which implies that the first term of the RHS of \eqref{eq:arithmetic} tends to $0$. Furthermore,
\begin{equation*}
\sum_{\substack{x\in \{0, 1\}^n : I(x)\in I_{\rm bad}\\ |I(x)| \le 2^{-n(h_2(\alpha)-\epsilon)}}} |I(x)| \le 2^m \cdot 2^{-n(h_2(\alpha)-\epsilon)},
\end{equation*}
since there are $2^m$ possible intervals $I(u)$ that can be intersected and the size of $I(x)$ is at most  $2^{-n(h_2(\alpha)-\epsilon)}$. Hence, by taking $m=n(h_2(\alpha)-2\epsilon)$, also the second term of the RHS of \eqref{eq:arithmetic} tends to $0$. This suffices to prove that the first property of Theorem \ref{th:chain} holds. 

Another possible solution is based on sparse graph codes, and the idea follows closely the scheme described in Section \ref{subsec:sparse}. Let $P \in {\mathbb F}_2^{m \times n}$ be the parity-check matrix of, e.g., an SC-LDPC code with sufficiently large degrees. Given $U^{1:m}$, we use it to initialize the values of the check nodes, and run belief-propagation guided decimation. Then, we define $g(U^{1:m})$ as the output of the algorithm. Given $X^{1:n}$, we define $f(X^{1:n})$ as $P X^{1:n}$. As pointed out in Section \ref{subsec:sparse}, $U^{1:m}$ differs from $f(g(U^{1:m}))$ in a vanishing fraction of positions, but we can cope with this issue by pre-coding $U^{1:m}$ with a negligible loss in rate. Note that this solution works well in practice, but, to the best of our knowledge, it is not provable.

Note that the second property of Theorem \ref{th:chain} is rather stringent. In fact, a weaker condition is sufficient, provided that the transmitter and the receiver have shared randomness. Let us describe this weaker condition in detail. Given a binary sequence $x^{1:n}\in {\mathbb F}_2^n$, let $\tau(x^{1:n})$ denote its type, i.e., the number of $1$s contained in the sequence. Then, for Theorem \ref{th:chain} to hold, rather than requiring that the distributions of $g(U^{1:m})$ and $X^{1:n}$ are roughly the same, it suffices that the distributions of $\tau(g(U^{1:m}))$ and $\tau(X^{1:n})$ are roughly the same and that a permuted version of $g(U^{1:m})$, call it $\pi(g(U^{1:m}))$, is transmitted over the channel. We require shared randomness, as the random permutation $\pi$ needs to be shared between the transmitted and the receiver. These concepts are formalized by the following proposition, whose proof immediately follows. 

\begin{proposition}\label{prop:TVdisttype}
Denote by $X^{1:n}$ a sequence of $n$ i.i.d. Bernoulli$(\alpha)$ random variables for some $\alpha\in [0, 1]$. Let $g(U^{1:m})$ be such that the total variation distance between the distribution of $\tau(g(U^{1:m}))$ and the distribution of $\tau(X^{1:n})$ is at most $\delta$. Let $\pi: [n]\to [n]$ be a random permutation. Then, the total variation distance between the distribution of $\pi(g(U^{1:m}))$ and the distribution of $X^{1:n}$ is at most $\delta$. 
\end{proposition}

\begin{proof}
Note that the type of a sequence is equal to the type of any permutation of such a sequence. By using this fact and the definition of type $\tau$, we have that
\begin{equation}\label{eq:TVchaineq}
\begin{split}
\frac{1}{2}\sum_{x\in \{0, 1\}^n}|{\mathbb P}_{\pi(g(U^{1:m}))}(x)-{\mathbb P}_{X^{1:n}}(x)| &=\frac{1}{2} \sum_{i=0}^n\sum_{\substack{x\in \{0, 1\}^n\\ \tau(x) = i}}|{\mathbb P}_{\pi(g(U^{1:m}))}(x)-{\mathbb P}_{X^{1:n}}(x)|\\
&=\frac{1}{2}\sum_{i=0}^n |{\mathbb P}_{\tau(\pi(g(U^{1:m})))}(i) - {\mathbb P}_{\tau(X^{1:n})}(i)|\\
&=\frac{1}{2}\sum_{i=0}^n |{\mathbb P}_{\tau(g(U^{1:m}))}(i) - {\mathbb P}_{\tau(X^{1:n})}(i)|\\
\end{split}
\end{equation}
On the one hand, the LHS of \eqref{eq:TVchaineq} represents the total variation distance between the distribution of $\pi(g(U^{1:m}))$ and the distribution of $X^{1:n}$. On the other hand, the RHS of \eqref{eq:TVchaineq} represents the total variation distance between the distribution of $\tau(g(U^{1:m}))$ and the distribution of $\tau(X^{1:n})$. Hence, the claim readily follows.
\end{proof}

As a result, by simply using an extra shared random permutation, if $g(U^{1:m})$ and $X^{1:n}$ have roughly the same type, then we fulfill the rather stringent second property of Theorem \ref{th:chain}.

In summary, in the source coding part of the scheme, we are given as input the vector $U^{1:m}$ with i.i.d. uniform distribution, and we generate the codeword $g(U^{1:m})$ that is transmitted over the channel. The distribution $g(U^{1:m})$ is close in total variation distance to an i.i.d. Bernoulli$(\alpha)$ distribution. Hence, by paying a negligible price in terms of the error probability, we can assume that the transmitted codeword is the sequence $X^{1:n}$ with i.i.d. Bernoulli$(\alpha)$ distribution. The {\em channel coding} part of the scheme consists in transmitting reliably $X^{1:n}$ over the channel. A similar problem has been considered in Section \ref{subsec:checks}. There, we have proved that $X^{1:n}$ can be reconstructed with high probability, given the channel output $Y^{1:n}$ and $nH(X\mid Y)$ additional bits of information. Recall that $h: \{0, 1\}^n \rightarrow \{0, 1\}^\ell$, with $\ell$ roughly equal to $nH(X\mid Y)$, hence the mapping $h(X^{1:n})$ provides these extra $nH(X\mid Y)$ bits. This means that, if we were able to share the vector $h(X^{1:n})$ between the transmitter and the receiver, we would be done. However, $h(X^{1:n})$ obviously depends on $X^{1:n}$, hence on the information vector $U^{1:m}$. Thus, it is not immediately clear how to share this vector. 

To solve this issue, we draw inspiration 
from the ``chaining'' construction introduced in \cite{BoM11, HRunipol}. 
We consider the transmission of $k$ blocks of information, and use a
part of the current block to store the parity-check vector of the previous
block. More specifically, in block $1$, we fill $U^{1:m}$ with information
bits, compute $X^{1:n} = g(U^{1:m})$, and $h(X^{1:n})$. In block $j$ ($j\in \{2, \cdots, k-1\}$), we fill $U^{1:m}$ with
the vector $h(X^{1:n})$ of the previous block, and store the 
information bits in the remaining positions. Note that the vector $h(X^{1:n})$ has roughly size $nH(X\mid Y)$, and recall that $m$ is close to $nh_2(\alpha)$. Hence, the transmission rate is approximately $h_2(\alpha)-H(X\mid Y)=C(W)$. Then, we compute $X^{1:n} =
g(U^{1:m})$ and $h(X^{1:n})$. In block $k$, we transmit
only the vector $h(X^{1:n})$ of the previous block at rate $C_{\rm s}(W)$, by using a code that achieves the symmetric capacity of $W$ (see Section \ref{subsec:symmcap} for polar coding and
sparse graph coding schemes that achieve such a goal). Note that, in the last block, we suffer a rate loss, as we are limited to a rate of $C_{\rm s}(W) < C(W)$. However, this
rate loss decays as $1/k$ and, by choosing $k$ large, we achieve
a rate arbitrarily close to $C(W)$.

At the receiver, we perform the decoding ``backwards'',
by starting with block $k$ and ending with block 1. More specifically, block $k$
can be easily decoded, as the underlying code achieves the symmetric
capacity of the channel. For block $j$ ($j\in \{k-1, \cdots, 1\}$),
the decoder can recover $X^{1:n}$ by using the channel output and $h(X^{1:n})$. Indeed, this last vector is obtained from the next block that is already decoded. Finally, from $X^{1:n}$
we deduce $U^{1:m} = f(X^{1:n})$.

Let us now describe how to construct a map $h$ that fulfills the desired property.
One possible solution is based on sparse graph codes. Let $P\in {\mathbb F}_2^{\ell\times n}$ be the parity-check matrix of an SC-LDPC code with sufficiently large degrees. Then, we set $h(X^{1:n}) = PX^{1:n}$, and conclude that the third property of Theorem \ref{th:chain} holds by following the argument of Section \ref{subsec:checks}. An alternative solution is based on polar coding techniques. We multiply $X^{1:n}$ with the polarizing matrix $G_n$ and extract the positions indexed by the complement of the set $\lset_{X\mid Y}$ defined in \eqref{eq:cardint1}. Then, we define $h(X^{1:n})$ as the result of this last operation. Furthermore, by applying \eqref{eq:cardint2}, we easily verify that the vector $h(X^{1:n})$ has the correct size, i.e., roughly $nH(X\mid Y)$. Furthermore, by definition, the set $\lset_{X\mid Y}$ contains the positions that are approximately a deterministic function of the previously decoded bits and the output. Hence, we can reconstruct $X^{1:n}$ with high probability given $Y^{1:n}$ and $h(X^{1:n})$, and the required property of $h$ holds.

Since we fill $U^{1:m}$ with the vector $h(X^{1:n})$ of the previous block, we require that $h(X^{1:n})$ has an i.i.d. uniform distribution. If this is not the case, we XOR the vector $h(X^{1:n})$ with an i.i.d. uniform sequence that is shared between the transmitter and the receiver. By doing so, the new vector of syndromes has exactly an i.i.d. uniform distribution. In the last two paragraphs of Section VI-A, we discuss how we can avoid this use of extra shared randomness at the cost of an increased error probability. 

\begin{figure}[t]

    \centering

\psset{arrowscale=1}
\psset{unit=0.5cm}
\psset{xunit=1,yunit=1}
\begin{pspicture}(-6,-10)(32,7)
\psframe(-5,3)(-1,4.5)
\psline[linecolor=black,linewidth=0.7pt]{->}(-1,3.75)(1,3.75)
\psframe(1,3)(6,4.5)
\psline[linecolor=black,linewidth=0.7pt]{->}(6,3.75)(10,3.75)
\psframe(10,3)(15,4.5)
\psline[linecolor=black,linewidth=0.7pt]{->}(3.5,3)(3.5,1.5)
\psframe(2.5,0)(4.5,1.5)

\rput(-3,5){\small{$u_{(1)}^{1:m}$}}
\rput(3.5,5){\small{$x_{(1)}^{1:n}$}}
\rput(0,4.2){\small{$g$}}
\rput(8,4.2){\small{$W$}}
\rput(4,2.25){\small{$h$}}
\rput(12.5,5){\small{$y_{(1)}^{1:n}$}}
\rput(0.75,0.75){\small{$s_{(1)}^{1:\ell}$}}

\rput(21.5,5){\small{$\hat{x}_{(1)}^{1:n}$}}
\rput(21.5,-1.5){\small{$\hat{x}_{(2)}^{1:n}$}}

\rput(28,5){\small{$\hat{u}_{(1)}^{1:m}$}}
\rput(28.5,-1.5){\small{$\hat{u}_{(2)}^{1:m}$}}

\psframe[linecolor=red, linewidth=1.5pt, linestyle=dashed](-4.95,-3.45)(-3,-2.05)
\psframe(-5,-3.5)(-1,-2)
\psline[linecolor=black,linewidth=0.7pt]{->}(-1,-2.75)(1,-2.75)
\psframe(1,-3.5)(6,-2)
\psline[linecolor=black,linewidth=0.7pt]{->}(6,-2.75)(10,-2.75)
\psframe(10,-3.5)(15,-2)
\psline[linecolor=black,linewidth=0.7pt]{->}(3.5,-3.5)(3.5,-5)
\psframe(2.5,-6.5)(4.5,-5)

\psline[linecolor=red,linewidth=0.7pt, linestyle=dashed]{->}(3.5,0.75)(-4,-2.75)

\rput(-3,-1.5){\small{$u_{(2)}^{1:m}$}}
\rput(3.5,-1.5){\small{$x_{(2)}^{1:n}$}}
\rput(0,-2.3){\small{$g$}}
\rput(8,-2.3){\small{$W$}}
\rput(4,-4.25){\small{$h$}}
\rput(12.5,-1.5){\small{$y_{(2)}^{1:n}$}}
\rput(0.75,-5.75){\small{$s_{(2)}^{1:\ell}$}}

\psframe[linecolor=red, linewidth=1.5pt, linestyle=dashed](-4.95,-9.95)(-3.05,-8.55)
\psframe(-5,-10)(-3,-8.5)
\rput(-4,-8){\small{$u_{(3)}^{1:\ell}$}}
\psframe(3,-10)(6,-8.5)
\psline[linecolor=black,linewidth=0.7pt]{->}(-3,-9.25)(3,-9.25)
\rput(4.6,-8){\small{$x_{(3)}^{1:n'}$}}
\psline[linecolor=black,linewidth=0.7pt]{->}(6,-9.25)(10,-9.25)
\rput(0,-8.8){\small{${\mathcal E}_{\rm s}$}}
\rput(8,-8.8){\small{$W$}}
\psframe(10,-10)(13,-8.5)
\rput(11.6,-8){\small{$y_{(3)}^{1:n'}$}}
\psline[linecolor=red,linewidth=0.7pt, linestyle=dashed]{->}(3.5,-5.75)(-4,-9.25)

\psline[linecolor=black,linewidth=0.7pt]{->}(13,-9.25)(19,-9.25)
\psline[linecolor=black,linewidth=0.7pt]{->}(15,-2.25)(19,-2.25)
\psline[linecolor=black,linewidth=0.7pt]{->}(15,4.25)(19, 4.25)

\psframe(19,-10)(21,-8.5)

\rput(16,-8.8){\small{${\mathcal D}_{\rm s}$}}

\psframe[linecolor=red, linewidth=1.5pt, linestyle=dashed](19.05,-9.95)(20.95,-8.55)

\psframe(19,-3.5)(24,-2)
\psframe(19,3)(24,4.5)

\psline[linecolor=black,linewidth=0.7pt]{->}(24,-2.75)(26,-2.75)
\psline[linecolor=black,linewidth=0.7pt]{->}(24,3.75)(26, 3.75)

\rput(18,-2.75){\small{${\mathcal D}$}}
\psline[linecolor=red,linewidth=0.7pt, linestyle=dashed]{->}(17,-3.25)(19,-3.25)
\psline[linecolor=red,linewidth=0.7pt, linestyle=dashed]{-}(17,-6.5)(17,-3.25)
\psline[linecolor=red,linewidth=0.7pt, linestyle=dashed]{-}(20,-6.5)(17,-6.5)
\psline[linecolor=red,linewidth=0.7pt, linestyle=dashed]{-}(20,-8.5)(20,-6.5)

\rput(18,3.75){\small{${\mathcal D}$}}
\psline[linecolor=red,linewidth=0.7pt, linestyle=dashed]{->}(17,3.25)(19,3.25)
\psline[linecolor=red,linewidth=0.7pt, linestyle=dashed]{-}(17,0)(17,3.25)
\psline[linecolor=red,linewidth=0.7pt, linestyle=dashed]{-}(27,0)(17,0)
\psline[linecolor=red,linewidth=0.7pt, linestyle=dashed]{-}(27,-2)(27,0)

\rput(24,0.6){\small{$\hat{s}_{(1)}^{1:\ell}$}}

\rput(19,-5.9){\small{$\hat{s}_{(2)}^{1:\ell}$}}
\rput(22,-9.25){\small{$\hat{u}_{(3)}^{1:\ell}$}}

\psframe(26,-3.5)(30,-2)
\psframe(26,3)(30,4.5)

\psframe[linecolor=red, linewidth=1.5pt, linestyle=dashed](26.05,-3.45)(27.95,-2.05)

\rput(25,4.2){\small{$f$}}
\rput(25,-2.3){\small{$f$}}

\end{pspicture}

\caption{Coding over asymmetric channels via the chaining construction. We consider the transmission of $k=3$ blocks, and store the vector of parity checks of block $j-1$ into block $j$ ($j\in \{2, 3\}$).}
\label{fig:chainingscheme}
\end{figure}
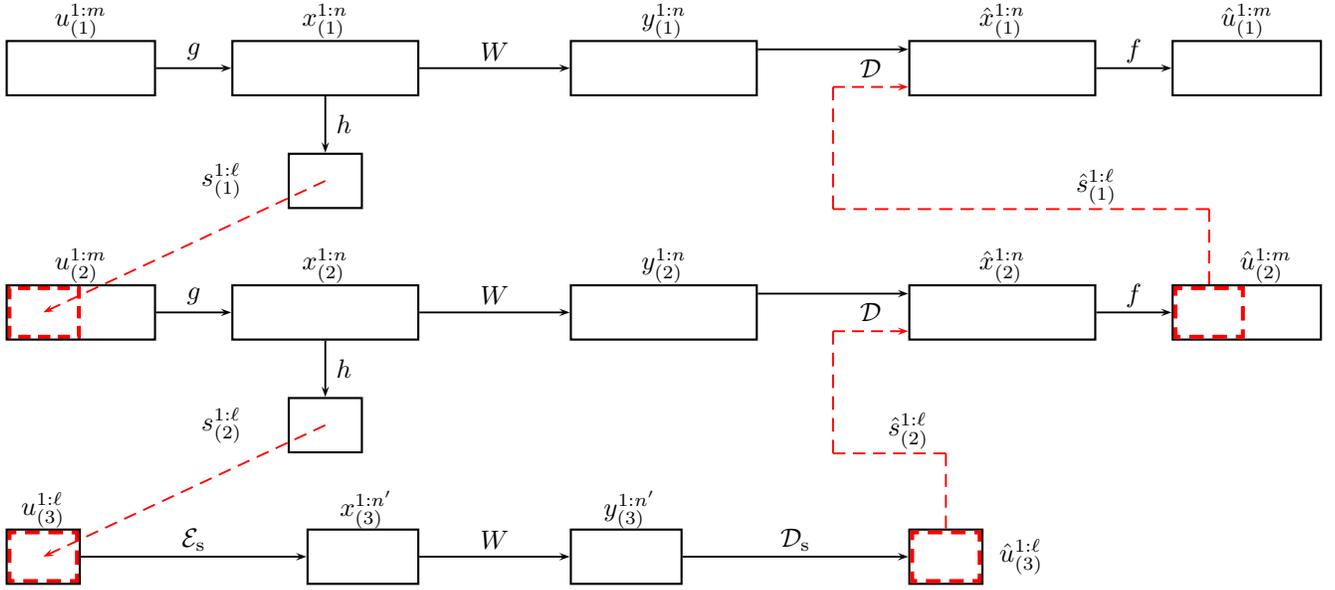

\vspace{1em}

\noindent {\bf Encoding.} Let ${\mathcal C}_{\rm s}$ be a code that achieves the symmetric
capacity of the channel $W$ and denote by ${\mathcal E}_{\rm s}$ its encoder. We consider the transmission of $k$ blocks and encode them in order, starting from block $1$ and ending with block $k$.

In block $1$, we place the information into $u_{(1)}^{1:m}$, compute the codeword $x_{(1)}^{1:n} = g(u_{(1)}^{1:m})$ and the vector $s_{(1)}^{1:\ell} = h(x_{(1)}^{1:n})$. The codeword $x_{(1)}^{1:n}$ is transmitted over the channel $W$ and the vector $s_{(1)}^{1:\ell}$ is stored into the next block. 

In block $j$ ($j\in \{2, \cdots, k-1\}$), we place $s_{(j-1)}^{1:\ell}$ and $nC(W)$ information bits into $u_{(j)}^{1:m}$. Note that this is possible, as $\ell \approx nH(X\mid Y)$, $m\approx nh_2(\alpha)$, and $H(X\mid Y)+C(W) = h_2(\alpha)$. Then, we compute the codeword $x_{(j)}^{1:n} = g(u_{(j)}^{1:m})$, and the vector $s_{(j)}^{1:\ell} = h(x_{(j)}^{1:n})$. Once again, the codeword $x_{(j)}^{1:n}$ is transmitted over the channel $W$ and the vector $s_{(j)}^{1:\ell)}$ is stored into the next block.

In block $k$, we place $s_{(k-1)}^{1:\ell}$ into $u_{(k)}^{1:\ell}$. Then, we map $u_{(k)}^{1:\ell}$ into the codeword $x_{(k)}^{1:n'}$ via the encoder ${\mathcal E}_{\rm s}$. Note that $n'$ is an integer roughly equal to $\ell/C_{\rm s}(W)$, as the code ${\mathcal C}_{\rm s}$ has rate close to $C_{\rm s}(W)$.

The overall rate of communication is given by
\begin{equation} \label{eq:ratechain}
R = \frac{nh_2(\alpha)+n(k-2)C(W)}{n(k-1)+nH(X\mid Y)/C_{\rm s}(W)},
\end{equation}
that, as $k$ goes large, tends to the required rate $C(W)$. 

\vspace{1em}

\noindent {\bf Decoding.} Denote by ${\mathcal D}_{\rm s}$ the decoder of the code ${\mathcal C}_{\rm s}$ and by ${\mathcal D}$ the decoder that recovers the codeword $x^{1:n}$ given the channel output $y^{1:n}$ and the vector $h(x^{1:n})$. The decoding process begins after all the $k$ blocks have been received, and it operates  ``backwards'', starting from block $k$ and ending with block $1$.  

In block $k$, the decoder ${\mathcal D}_{\rm s}$ accepts as input the received message $y_{(k)}^{1:n'}$. The output is the estimate $\hat{u}_{(k)}^{1:\ell}$ on the payload $u_{(k)}^{1:\ell}$ of block $k$. This immediately yields the estimate $\hat{s}_{(k-1)}^{1:\ell}$ on the vector $s_{(k-1)}^{1:\ell}$ of block $k-1$.

In block $j$ ($j\in \{k-1, \cdots, 2\}$), the decoder ${\mathcal D}$ accepts as inputs the received message $y_{(j)}^{1:n}$ and the previously obtained estimate $\hat{s}_{(j)}^{1:\ell}$. The output is the estimate $\hat{x}_{(j)}^{1:n}$ on the codeword $x_{(j)}^{1:n}$ of block $j$. Then, we compute $\hat{u}_{(j)}^{1:m} = f(\hat{x}_{(j)}^{1:n})$, which yields an estimate on the information bits transmitted in block $j$ and on the vector $s_{(j-1)}^{1:\ell}$ of block $j-1$.

For block $1$ the decoding process is the same as that for block $j$ ($j\in \{k-1, \cdots, 2\}$). The only difference consists in the fact that $\hat{u}_{(1)}^{1:m}$ contains solely an estimate on information bits. 

The situation is schematically represented in Figure \ref{fig:chainingscheme}. 

\vspace{1em}

\noindent {\bf Performance.} There are four possible types of errors.
\begin{enumerate}
	\item In block $j$ ($j\in \{1, \cdots, k-1\}$), given that $\hat{x}_{(j)}^{1:n}= x_{(j)}^{1:n}$, we might have that $\hat{u}_{(j)}^{1:m} \neq u_{(j)}^{1:m}$.  
	\item In block $j$ ($j\in \{1, \cdots, k-1\}$), the encoder might fail to produce a codeword $x_{(j)}^{1:n}$ with the correct distribution (namely, with roughly $n\alpha$ $1$s).
	\item In block $j$ ($j\in \{1, \cdots, k-1\}$), we might have that $\hat{x}_{(j)}^{1:n}\neq x_{(j)}^{1:n}$.
	\item In block $k$, we might have that $\hat{x}_{(k)}^{1:n}\neq x_{(k)}^{1:n}$.
\end{enumerate}

By hypothesis, the map $f$ fulfills the first property stated in Theorem \ref{th:chain}. Hence, by the union bound, the probability that the first error event takes place is at most $(k-1)\delta$. Similarly, the map $g$ fulfills the second property stated in Theorem \ref{th:chain}. Hence, by the union bound, the probability that the second error event takes place is at most $(k-1)\delta$. Furthermore, the map $h$ fulfills the third property stated in Theorem \ref{th:chain}. Hence, by the union bound, the probability that the third error event takes place is at most $(k-1)\delta$. As the code ${\mathcal C}_{\rm s}$ achieves the symmetric capacity of the channel $W$, the last event occurs with probability at most $\delta$. As a result, the error probability of the proposed scheme is upper bounded by $3k\delta$. Recall that $k$ is large but fixed (it only depends on the rate we want to achieve), hence, by choosing $\delta$ sufficiently small, the proof of Theorem \ref{th:chain} is complete.

\section{Performance Comparison Between the Three Paradigms} \label{sec:comp}


\subsection{Error Probability}

First, consider \emph{Gallager's mapping}. Recall that in Section \ref{sec:gallager} we describe two schemes: one based on a single non-binary code $\mathcal C$, and the other based on $t$ binary codes $\{\mathcal C_j\}_{j\in [t]}$. Let $W$ be the transmission channel and let $W'$ be defined as in \eqref{eq:mapqex}. Then, for the scheme based on a single non-binary code, the error probability is the same as that of the transmission of $\mathcal C$ over $W'$. For $j\in [t]$, let $W_j''$ be defined as in \eqref{eq:defWg}. Then, for the scheme based on $t$ binary codes, the error probability is upper bounded by the sum over $j$ of the error probabilities of the transmission of $\mathcal C_j$ over $W_j''$. This means that we need to multiply the error probability by a factor of $t$. 

Second, consider the \emph{integrated scheme}. In \cite{HY13} the authors provide a comparison between the second-order error exponent of Gallager's mapping and of the integrated scheme with polar codes revised in Section \ref{subsec:polar}. In particular, let $p$ be the input distribution induced by Gallager's mapping. Then, if the transmission rate $R$ is sufficiently close\footnote{The exact condition is $R> I(p)-I(p^*)(I(p^*)-I(p))$, where $p^*$ is the capacity-achieving input distribution.} to $I(p)$, the integrated scheme achieves a better second-order error exponent than Gallager's mapping.

Third, consider the \emph{chaining construction}. Recall that in Section \ref{sec:chaining} we have divided the transmission in $k$ blocks and performed the decoding ``backwards''. This method suffers from error propagation, in the sense that an error occurring in block $t$ propagates to all the previous blocks from $t-1$ to $1$. Hence, we need to multiply the error probability by a factor of $k$, as pointed out at the end of Section \ref{sec:chaining}. Note that such a behavior occurs also in \cite{MHSU14}, where a similar chaining construction is employed to devise polar coding schemes for the broadcast channel. More specifically, in formulae (41) and (51) of \cite{MHSU14} there is a factor of $k$ in the expression of the error probability.

Recall also that the total variation distance between the distribution of $g(U^{1:m})$ and an i.i.d. Bernoulli$(\alpha)$ distribution is at most $\delta$. Thus, the total variation distance between the distribution of $h(X^{1:n})$ and i.i.d. uniform distribution is at most $\delta$.\footnote{To see this, let $X$ and $X'$ be random vectors such that their total variation distance is at most $\delta$. This means that we can couple them so that they differ with probability at most $\delta$. Hence, also $h(X)$ and $h(X')$ differ with probability at most $\delta$, which implies that the total variation distance between $h(X)$ and $h(X')$ is at most $\delta$.} Furthermore, if the total variation distance between the distribution of $U^{1:m}$ and an i.i.d. uniform  distribution is at most $\delta$, then the total variation distance between the distribution of $g(U^{1:m})$ and an i.i.d. Bernoulli$(\alpha)$ distribution is at most $2\cdot\delta$. As a result, in block $t$ the total variation distance between the distribution of the codeword and an i.i.d. Bernoulli$(\alpha)$ distribution is at most $t\cdot\delta$. Since the overall error probability is upper bounded by the sum of the error probabilities of the single blocks, the dependence of the error probability is quadratic in $k$.

As previously pointed out, this issue can be solved (and the dependence of the error probability can be made again linear in $k$) by requiring common randomness. Anyway, in the case of polar codes, the fact that the dependency of the error probability is linear or quadratic in $k$ does not influence much the scaling behavior in the regime in which a rate $R<C(W)$ is fixed and $n$ goes to $\infty$. Indeed, in this regime the error probability under successive cancellation decoding scales as $2^{-\sqrt{n}}$, and the number of blocks $k$ is a constant independent of $n$.

\subsection{Rate Penalty}

First, consider \emph{Gallager's mapping}. In this case, the rate penalty comes from the fact that the distribution $p$ induced by the map might not be exactly equal to the capacity-achieving input distribution $p^*$. We quantify the rate penalty in Proposition \ref{prop:boundmi} in terms of the total variation distance between $p$ and $p^*$, and of the cardinalities of the input and output alphabets. The rate penalty is also studied in \cite{HY13} for the special case of binary input alphabet. Note that the bound obtained in formula (25) of \cite{HY13} is tighter\footnote{Let $\delta$ be the total variation distance between $p$ and $p^*$. Then, formula (25) of \cite{HY13} gives that the rate penalty is $O(\delta^2)$, whereas Proposition \ref{prop:boundmi} gives that the rate penalty is $O(\delta \log(1/\delta))$.} than our bound of Proposition \ref{prop:boundmi}, but it is significantly less general as it crucially uses the fact that the input alphabet is binary.

Second, consider the \emph{integrated scheme}. In this case, the rate penalty comes from the fact that we need to pre-code the syndrome vector, as the output of the belief-propagation guided decimation algorithm does not coincide exactly with the given syndrome vector. In Section \ref{sec:integrated}, we have observed that the fraction of unfulfilled syndromes tends $0$ as $n$ goes large. Hence, also the rate penalty can be made arbitrarily small. However, the rigorous proof of these statements remains an open problem. 

Third, consider the \emph{chaining construction}. In this case, the rate penalty comes from the fact that the rate in the last block is $C_{\rm s}(W)<C(W)$, where $C(W)$ and $C_{\rm s}(W)$ denote the capacity and the symmetric capacity of the channel $W$, respectively. From formula \eqref{eq:ratechain}, we immediately obtain that the rate penalty is $O(1/k)$, where $k$ is the number of blocks.


\subsection{Computational Complexity}

First, consider \emph{Gallager's mapping}. For the scheme based on a single non-binary code, the computational complexity of the non-linear mapper scales as a linear function of the cardinality of the domain of the map. For the scheme based on $t$ binary codes, the computational complexity is linear in $t$. Let us briefly discuss why this is the case. First, assume that the input alphabet $\mathcal X$ of the channel is binary, i.e., $\mathcal X = \{0, 1\}$. Without loss of generality, assume that $p^*(0)= \bar{\alpha}$ for some $\bar{\alpha}<1/2$. Indeed, if $p^*(0)\ge 1/2$, we repeat the same argument with the roles of the input symbol $0$ and of the input symbol $1$ exchanged. Define $2^{-\bar{t}}$ as the largest power of $2$ that is smaller than $\bar{\alpha}$. Then, all the sequences of $t$ bits starting with $\bar{t}$ $0$s are assigned to the input symbol $0$. The remaining sequences, except those that start with $\bar{t}$ $1$s, are assigned to the input symbol $1$. The sequences that start with $\bar{t}$ $1$s are assigned in part to the input symbol $0$ and in part to the input symbol $1$. To do so, we simply rescale the probabilities and iterate the procedure above. Clearly, the complexity of this algorithm is linear in $t$. Consider now the case in which $\mathcal X$ is not binary and divide the $(0, 1)$ interval into $|\mathcal X|$ sub-intervals. Denote by $I_x$ the sub-interval of length $p^*(x)$ that corresponds to $x\in \mathcal X$. Without loss of generality, assume that $\arg\min_x p^*(x)=0$. Indeed, we can always re-label the input symbols. Let $p^*(0)=\bar{\alpha}$ and define $2^{-\bar{t}}$ as the largest power of $2$ that is smaller than $\bar{\alpha}$. Partition the interval $(0, 1)$ into $2^{\bar{t}}$ parts of the same size. If a sub-interval of size $2^{-\bar{t}}$ is included into $I_x$ for some $x\in \mathcal X$, then all the associated sequences of $t$ bits are assigned to $x$. If not, these sequences are split between two different input symbols and we are back to the case of a binary input alphabet. As a result, the complexity of this algorithm is linear in $t$ and linear in $|\mathcal X|$. Furthermore, note that the total variation distance between $p$ and $p^*$, call it $\delta$, scales with $2^{-t}$. Hence, the computational complexity scales as $\log (1/\delta)$.

Second, consider the \emph{integrated scheme}. This approach has the same computational complexity as the standard channel coding solution for the transmission over a symmetric channel.

Third, consider the \emph{chaining construction}. The computational complexity of this scheme scales as a linear function of the number of blocks $k$, as there are $k$ blocks to be encoded and decoded.



\subsection{Trade-off} \label{sc-tradeoff}

One way to compare the various schemes is to run simulations and
compare the resulting parameters (error probability, rate penalty and
computational complexity). For example, we could fix the
computational complexity and block length and plot the error
probability as a function of the rate for a specific channel.
Unfortunately, such a comparison is much more difficult and less
meaningful than it might seem at first. All the schemes discussed
are by definition ``low-complexity'' and there are myriads of ways
of improving each one of them with no clear way of evaluating the
computational complexity. 

As an example, think of polar codes.  Successive cancellation is
the standard decoding algorithm but, in order to get good performance,
it is typically necessary to strengthen the code and apply various
forms of list decoding \cite{TVa15}. The number of variants and
enhancements that have been proposed to date is vast and many of
those have a significant impact on the final result. Similar remarks
apply to the other components as well. As a result, it is hard to carry
out a fair comparison by numerical simulations.

We therefore take a different route here. Rather than using simulations
to gain insight, we use ``scaling'' relationships between the various
parameters to bring out some of the basic trade-offs.  The advantage
of this approach is that it does not depend on implementation
details, hence it makes it easier to see the broad picture.  The
disadvantage of this approach is that it hides some of the involved
constants. Furthermore, scaling relationships are only known for
pairs of parameters (e.g., the block length versus the gap to capacity)
while the remaining parameters (e.g., the error probability) are
fixed.

Let us start by introducing some notation and by reviewing these
basic scaling relationships.

\begin{itemize}
\item Let $\mu_{\rm cha}$ denote the \emph{scaling exponent} for
the \emph{channel coding} problem. In words, $\mu_{\rm cha}$
is the parameter so that to first order the block length of the
channel code is equal to $(1/\delta)^{\mu_{\rm cha}}$, where $\delta
= C(W)-R$ is the {\em gap to capacity}. As mentioned in the
introduction, for polar codes we have $3.579\le \mu_{\rm cha}
\le 4.714$ \cite{HAU14, GB14, MHU15}. For SC-LDPC codes, no rigorous
result about the scaling exponent is known, but a simple heuristic
argument suggests that $\mu_{\rm cha}\approx 3$. 

Let us briefly explain this heuristic argument. Recall that in
SC-LDPC codes we couple let's say $k$ codes together. There are two
effects that limit the finite-length performance. First, the block
length of each of the composite codes acts as the ``effective''
block length of the coupled code. Hence, we cannot approach capacity
closer than what this effective block length dictates. Second, at
the boundary we have a fixed rate loss and this rate loss is amortized
over the $k$ blocks, i.e., the amortized loss is of the order $1/k$.
Let $n$ be the block length of the overall code, which is obtained
by coupling $n^{\alpha}$ codes of effective block length $n^{1-\alpha}$,
for some $\alpha\in (0, 1)$. As a result, the loss due to the
boundary effects is of the order $1/n^{\alpha}$. Furthermore, the
best codes have a loss that scales as the square root of their block
length. This gives us a loss of $1/n^{(1-\alpha)/2}$. The overall
loss is the sum of these two components. We get the smallest such
sum if we choose $\alpha$ so that both losses are equal.  By setting
\begin{equation}
\frac{1}{n^\alpha} = \frac{1}{n^{(1-\alpha)/2}},
\end{equation}
we obtain $\alpha = 1/3$. Hence, each of the composite blocks should
have a length of $n^{2/3}$ and we should couple $n^{1/3}$ of them.
This gives an overall loss of the order of $n^{-1/3}$. In other words, the scaling
behavior is $n \sim (1/\delta)^3$.

\item Similarly, let $\mu_{\rm src}$ denote the \emph{scaling
exponent} for the lossless \emph{source coding} problem. Then, the
block length of the source code scales as $(1/\delta)^{\mu_{\rm
src}}$, where $\delta$ is the number of excess bits that our scheme
requires (normalized by the block length) compared to the theoretical
limit $H(X)$.

\item Let $\Xi_{\rm dec}(\delta)$ denote the decoding complexity
measured as the number of binary operations that are required {\em
per transmitted bit}.  This quantity scales as $(1/\delta)^a\cdot
\log^b(1/\delta)$.  For polar codes, we have  $a=0$ and $b=3$; and
for SC-LDPC codes, we have $a=1$ and $b=2$. Let us discuss these
expressions in more detail. 

This scaling is easily seen as follows. For polar codes, the number
of real operations needed to decode a whole block of length $n$
scales like $n \log n$.  Hence the number of real operations needed
per transmitted bit scales like $\log n$, and $n$ itself scales
like $(1/\delta)^{\rm cha}$, as we have discussed.  Hence, measured
in real operations, we get a complexity of $\log(1/\delta)$. Now,
let us compute the number of {\em binary} operations that are
required.  As we approach capacity, we need to refine the quantization
levels of the decoding operations.  In particular, we need
$\log(1/\delta)$ bits of accuracy in the operations and the complexity
of a basic operation is quadratic in the number of bits \cite{hassani2012polar}.  As a
result, the complexity measured in binary operations scales like
$\log^3(1/\delta)$ for polar codes. 

Similar considerations apply also to SC-LDPC codes. We need $1/\delta$
iterations and, as for polar codes, in each iteration one basic
operation requires $\log^2(1/\delta)$ binary operations.

\item Let $\Xi_{\rm enc}(\delta)$ denote the encoding
complexity. For polar codes, the operations are essentially the same
as for the channel coding problem. Hence, the same scaling
law holds. On the contrary, no scaling law is known for SC-LDPC codes.
\end{itemize}

First, consider the \emph{integrated scheme}. The total gap to
capacity $\delta$ is due to finite-length effects and it can be
decomposed into two contributions: $\delta_{\rm src}$ is the gap
relative to the source coding part, and $\delta_{\rm cha}$ is the
gap relative to the channel coding part. For the scheme that uses
sparse-graph codes, the gap relative to the source coding part is
due to the fact that we need to pre-code the information bits.
According to the numerical simulations in \cite{AMV15, KVNP14}, we
observe that $\delta_{\rm src}\to 0$ as $n\to \infty$. However, no
rigorous result is known as concerns the scaling of $\delta_{\rm
src}$. For the scheme that uses polar codes, the overall gap to
capacity $\delta$ has the same scaling as for the symmetric channel
coding problem, as proved in \cite{fong2016scaling}. The total
complexity $\Xi_{\rm tot}$ can be decomposed into two contributions:
the encoding complexity $\Xi_{\rm enc}$, and the decoding complexity
$\Xi_{\rm dec}$. Finally, the block length has to be chosen
sufficiently large so that both the source and the channel coding
task can be performed within the desired gaps. This discussion is
summarized by the following expressions:
\begin{equation}\label{eq:scalint}
\begin{split}
\delta&=\delta_{\rm src}+\delta_{\rm cha},\\
\Xi_{\rm tot} &= \Xi_{\rm enc}(\delta_{\rm src}) + \Xi_{\rm dec}(\delta_{\rm cha}), \\
n &= \max\left(\frac{1}{\delta_{\rm src}^{\mu_{\rm src}}},\frac{1}{\delta_{\rm cha}^{\mu_{\rm cha}}}\right).
\end{split}
\end{equation}

Second, consider the \emph{chaining construction}. Recall that we
divide the transmission into $k$ blocks.  In the first $k-1$ blocks,
we transmit close to the capacity of the asymmetric channel. In
these blocks we have two losses, both due to finite-length effects.
Let $\delta_{\rm{src}}$ be the loss due to the source coding part
and $\delta_{\rm{cha}}$ be the loss due the channel coding part.
In the last block, we transmit close to the symmetric capacity of
the channel, which by assumption is strictly smaller than the actual
capacity. Such a loss is amortized over the $k$ blocks.  Let the
amortized gap be denoted by $\delta_{\rm{xtr}}$. Then, we have that
$\delta_{\rm{xtr}}=1/k$. The scaling laws for this scheme can be
summarized as follows:
\begin{equation}\label{eq:scalchain}
\begin{split}
\delta&=\delta_{\rm{src}}+\delta_{\rm{cha}}+\delta_{\rm{xtr}},\\
\Xi_{\rm tot} &= \Xi_{\rm enc}(\delta_{\rm{src}}) + \Xi_{\rm dec}(\delta_{\rm{cha}}),  \\
n &= \frac{1}{\delta_{\rm{xtr}}} \max\left(\frac{1}{\delta_{\rm{src}}^{\mu_{\rm src}}},\frac{1}{\delta_{\rm{cha}}^{\mu_{\rm cha}}}\right).
\end{split}
\end{equation}
By assumption, the total gap to capacity is the sum of the three
contributions $\delta_{\rm{src}}$, $\delta_{\rm{cha}}$ and
$\delta_{\rm{xtr}}$. The total complexity $\Xi_{\rm tot}$ is equal
to the sum of the source encoding and the channel decoding complexities,
which are functions of their respective gaps. Finally, let us discuss
the scaling of the block length. There are $k$ blocks, where
$k=1/\delta_{\rm{xtr}}$; and the length of each block has to be
chosen sufficiently large so that both the source and the channel
coding task can be performed within the desired gaps.

If we compare the expressions \eqref{eq:scalint} for the integrated
scheme with the expressions \eqref{eq:scalchain} for the chaining
construction, we see that they are essentially the same, except
that the block length for the chaining construction contains an extra
factor $1/\delta_{\rm xtr}$. Viewed like this, the integrated scheme
seems a clear winner. However, recall that in the integrated scheme
we are limited to polar codes (if we want provable results), and
their scaling behavior is markedly suboptimal. On the contrary, for
the chaining construction we have more degrees of freedom and we
can potentially pick schemes with a better scaling behavior.

Third, consider \emph{Gallager's mapping}. Let us focus on the
solution with binary codes, and let $t$ be the number of such codes.
The total gap to capacity $\delta$ can be decomposed into two
contributions: $\delta_{\rm{xtr}}$ is due to the fact that we select
a rational approximation of the capacity-achieving distribution,
and $\delta_{\rm cha}$ is the usual gap due to finite-length effects
for the channel coding part. By Proposition \ref{prop:boundmi}, we
have $\delta_{\rm xtr} = 2^{-t}$, where we consider only the
leading term that is exponential in $t$ and neglect the lower order
terms that are linear in $t$. The total complexity $\Xi_{\rm tot}$
can be decomposed into two contributions: the encoding complexity
$\Xi_{\rm enc}$, which accounts for the construction of the non-linear
mapper and scales as $\log(1/\delta_{\rm xtr})$; and the decoding complexity $\Xi_{\rm dec}$. Finally, the block length has
to be chosen sufficiently large so that the channel coding task can
be performed within the desired gap. Note that the rates of each
of the $t$ codes tend to $0$, as their sum is close to the capacity
of the channel. Hence, the characterization of the scaling of the
block length does not follow from standard considerations, and it
remains an open problem. For this reason, we do not state an explicit
formula for the required block length, and the discussion is
summarized by the following expressions:
\begin{equation}
\begin{split}
\delta&=\delta_{\rm cha}+\delta_{\rm xtr},\\
\Xi_{\rm tot} &= \log\left(\frac{1}{\delta_{\rm xtr}}\right)+\Xi_{\rm dec}(\delta_{\rm cha}). \\
\end{split}
\end{equation}

\subsection{Universality}

We say that a coding scheme achieves capacity \emph{universally}
over a class of channels if it achieves the capacity of each channel
in the class at the same time. This means that the coding scheme
is not tailored to the specific channel, rather it can be used for
transmission over any channel in the class. SC-LDPC codes universally
achieve capacity over the class of B-DMCs~\cite{KRU13}. Polar codes,
on the contrary, are not universal, as the sets defined in
\eqref{eq:cardint1} depend on the transmission channel. Therefore,
several ``polar-like'' schemes have been developed to solve this
issue~\cite{HRunipol,SaWa16}. Note that by using the techniques in
\cite{HRunipol,SaWa16} we can achieve the actual capacity of the
channel as opposed to the compound capacity\footnote{Recall that for the transmission over asymmetric channels the compound capacity may be strictly smaller than the capacity of the individual channels, as the capacity-achieving input distribution varies with the channel.}.

First, consider \emph{Gallager's mapping}. This scheme is not universal, as different transmission channels require different capacity-achieving distributions, hence different mappings. On the contrary, the \emph{integrated scheme} and the \emph{chaining construction} are universal, provided that the underlying component codes are universal (e.g., we use either SC-LDPC codes or the ``polar-like'' schemes described in \cite{HRunipol,SaWa16}).

\subsection{Common Randomness}

First, consider \emph{Gallager's mapping} and let $W'$ be defined as in \eqref{eq:mapqex}. Then, if we use polar codes in order to achieve the symmetric capacity of $W'$, the values of the frozen bits have to be chosen randomly and shared between the transmitter and the receiver. As pointed out in \cite{HY13}, there exists at least one random seed that achieves an error probability no larger than the expected error probability (which is averaged over all possible choices of the frozen bits). Furthermore, at least a fraction $1-\gamma$ of the seeds achieves an error probability that is at most a factor $1/\gamma$ larger than the average. It remains an interesting open problem to establish whether we can still set the frozen bits to $0$, as in the binary symmetric case. 

Second, consider the \emph{integrated scheme}. This approach requires common randomness both in the version based on polar codes and in the version based on sparse graph codes. More specifically, in the polar version described in Section \ref{subsec:polar}, we fill the positions in $\mathcal F_{\rm r}$ with a sequence chosen uniformly at random and shared between the transmitted and the receiver. Furthermore, we encode the positions in $\mathcal F_{\rm d}$ via the random map defined in \eqref{eq:randround} that also needs to be shared. In the sparse graph version described in Section \ref{subsec:sparse}, the syndrome vector $S_2^{1:nH(X\mid Y)}$ is chosen uniformly at random and shared between the transmitter and the receiver.

Third, consider the \emph{chaining construction} and assume that
the three properties in the hypothesis of Theorem \ref{th:chain}
hold. By separating the source coding and channel coding parts of
the scheme, this approach does not require common randomness, hence
it can be interpreted as a derandomized version of the integrated
scheme. This establishes another connection between information
theory and the theory of derandomizing algorithms. Several applications
of derandomization to coding theory can be found in \cite{Mahdi10thesis},
i.e., information-theoretically secure schemes for the wiretap
channel, nearly optimal explicit measurement schemes for combinatorial
group testing, design of ensembles of capacity achieving codes, and
construction of codes arbitrarily close to the Gilbert-Varshamov
bound. Furthermore, the link between polarization and randomness
extraction is investigated in \cite{Ab15}, where applications to
the Slepian-Wolf problem and to secret key generation are provided.

Assume that we want to substitute the stringent condition
on the distance between the distributions of $g(U^{1:m})$ and
$X^{1:n}$ with the relaxed condition on the distance between the
distributions of their types. Then, as detailed in Proposition
\ref{prop:TVdisttype}, the transmitter and the receiver need to
share $k$ random permutations, where $k$ is the number of blocks
used for the transmission. We will now show that no shared
randomness is in fact necessary.

As a starting point, recall that the error probability under
the stringent condition with no random permutation is the same as the error
probability under the relaxed condition with random permutations.
Furthermore, this probability is upper bounded by $3 k \delta$.

The error probability is an average over all channel realizations
and all permutations. Hence, for any $\gamma > 0$, at least a fraction $1-\gamma$ of the permutations have an error probability of at most $3 k \delta/\gamma$. By picking $\gamma=\sqrt{3 k \delta}$, we have that, with probability at least $1-\sqrt{3 k \delta}$, a randomly chosen permutation
has an error probability of at most $\sqrt{3 k \delta}$. Hence, no shared randomness is needed, as a fixed set of $k$ permutations will work with high probability.

By using this same argument, we can also show that no shared randomness is needed in the integrated scheme, as a fixed set of bits for the positions in $\fset_{\rm r}$ and a fixed vector for the syndrome $S_2^{1:nH(X\mid Y)}$ will work with high probability.  

\section{Concluding Remarks}\label{sec:concl}

This paper discusses and compares the performance of three different paradigms to achieve the capacity of \emph{asymmetric} DMCs.

The first approach is based on \emph{Gallager's mapping}. The idea was first described in \cite{Gal68}, and it consists of employing a non-linear function in order to make the input distribution match the capacity-achieving one. In this way, we can achieve the capacity of asymmetric channels by using either $q$-ary or binary codes that are capacity-achieving for suitably defined symmetric channels. 

The second approach consists in an \emph{integrated scheme} that simultaneously performs the tasks of source coding and of channel coding. The idea was first presented for polar codes in \cite{HY13}, and here we extend it to sparse graph codes. Indeed, sparse graph codes can be effectively used to create
biased codewords from uniform bits (source coding part) and to provide error correction (channel coding part). Given the vector of syndromes, we generate the codeword by running a belief-propagation algorithm with decimation steps. This technique works well in practice, but the proof that the scheme is capacity-achieving remains an open problem.

The third approach consists in a \emph{chaining construction}, where we consider the transmission of $k$ blocks and use a part of the current block to store the syndromes coming from the previous block. The idea was first proposed in \cite{BoM11}, and here we show how to use it to provably achieve the capacity of asymmetric DMCs. By decoupling completely the source coding from the channel coding task, we can employ an optimal scheme to reach each of these two objectives separately. Thus, many combinations are possible: for example, we can use polar codes or homophonic codes for the source coding part, and polar codes or spatially coupled codes for the channel coding part. 


As for the integrated scheme and the chaining construction, we restrict our discussion to the case of binary-input channels. In order to extend our results to channels with an arbitrary finite input alphabet, we require schemes that solve the source coding and the channel coding tasks in the non-binary case. For the source coding part, several works have focused on the construction of polar codes for arbitrary input alphabets \cite{STA09, MT10, PB13, SaP13, NT16, Nas17a, Nas17b}. Furthermore, we can also easily generalize the solution based on homophonic coding to non-binary alphabets (the interval algorithm of \cite{HoHa01} works directly in the non-binary case). For the channel coding part, recall that in Section~\ref{sec:gallager} we have converted a non-binary channel into several binary channels by using the chain rule of mutual information (see formula \eqref{eq:chainrulegall}). Here, the same idea can be applied as well. Alternatively, we can use directly non-binary spatially coupled codes \cite{UKS11, PAC13, AA16, WKMFC14} or non-binary polar codes \cite{STA09, MT10, PB13, SaP13, NT16, Nas17a, Nas17b}.

\section*{Acknowledgment}
We would like to thank E. Telatar for helpful discussion. Furthermore, we would like to thank the editor Prof. H. D. Pfister and the two anonymous reviewers for the many comments that helped us to significantly improve the quality of this paper. This work was supported by grant No. 200020\_146832/1 of the Swiss National Science Foundation.

\appendix

\subsection{Proof of Propositions in Section ~\ref{subsec:symmcap}}\label{app:proofsecII1}

\begin{proof}[Proof of Proposition~\ref{prop:symm}]
By definition, we have that
\begin{equation*}
\begin{split}
\Ldens{a}^+ (y) \Delta y &\approx \int\limits_{t \in L^{-1}([y, y+\Delta y])} W(t\mid 1) \,\,dt = \int\limits_{t \in L^{-1}([y, y+\Delta y])} e^{L(t)}W(t\mid -1) \,\,dt \\
& \approx e^y\int\limits_{t \in L^{-1}([y, y+\Delta y])} W(t\mid -1) \,\,dt = e^y \Ldens{a}^- (y) \Delta y,
\end{split}
\end{equation*}
where $L^{-1}$ is the inverse of the log-likelihood ratio defined in~\eqref{eq:loglike}. By taking $\Delta y\to 0$, we obtain that 
\begin{equation}\label{eq:reexp}
\Ldens{a}^+ (y)=e^y\Ldens{a}^- (y).
\end{equation} 
With the change of variable $y\to -y$, we also obtain that
\begin{equation}\label{eq:reexp2}
\Ldens{a}^- (-y)=e^y\Ldens{a}^+ (-y).
\end{equation} 
As a result, condition \eqref{eq:symm} is fulfilled for the $L$-density $\Ldens{a}^{\rm s}(y)$ defined in \eqref{eq:ldenssymm} and the statement follows. 
\end{proof}

\begin{proof}[Proof of Proposition~\ref{prop:cap}]
Since the log-likelihood ratio constitutes a sufficient statistic, two B-DMCs are equivalent if they have the same $L$-densities given that $X=\pm 1$ is transmitted. As a representative for the equivalence class, we can take 
\begin{equation}\label{eq:defcap}
\begin{split}
W(y\mid 1) &= \Ldens{a}^+(y),\\ 
W(y\mid -1) &= \Ldens{a}^-(y).\\ 
\end{split}
\end{equation} 
By definition of log-likelihood ratio and by using \eqref{eq:reexp}, we have
\begin{equation*}\label{eq:loglikedens}
L(y) = \ln \frac{W(y\mid 1)}{W(y\mid -1)} = \ln \frac{\Ldens{a}^+(y)}{\Ldens{a}^-(y)} = y. 
\end{equation*}
Therefore,
\begin{equation*}
\lim_{\Delta y\to 0} \frac{{\mathbb P}(L(Y)\in [y, y+\Delta y]\mid X = \pm 1)}{\Delta y} = \Ldens{a}^{\pm}(y),
\end{equation*}
which means that \eqref{eq:defcap} is a valid choice. 

Let $X$ be uniformly distributed. Then, after some calculations we have that
\begin{equation}\label{eq:pr1cap}
\begin{split}
C_{\rm s}(W) = H(Y) - H(Y\mid X) &= \frac{1}{2} \int  W(y\mid 1) \log \frac{2W(y\mid 1)}{W(y\mid 1)+W(y\mid -1)}\,\,dy \\
&+\frac{1}{2} \int  W(y\mid -1) \log \frac{2W(y\mid -1)}{W(y\mid 1)+W(y\mid -1)}\,\,dy.
\end{split}
\end{equation}
By applying \eqref{eq:defcap} and \eqref{eq:reexp}, the first integral simplifies to 
\begin{equation}\label{eq:pr1cap2}
\frac{1}{2} \int \Ldens{a}^+(y) \left(1-\log (1+e^{-y})\right)\,\,dy.
\end{equation}
By applying \eqref{eq:defcap}, doing the change of variables $y\to -y$ and using \eqref{eq:reexp2}, the second integral simplifies to 
\begin{equation}\label{eq:pr1cap3}
\frac{1}{2} \int \Ldens{a}^-(-y) \left(1-\log (1+e^{-y})\right)\,\,dy.
\end{equation}
By combining \eqref{eq:pr1cap}, \eqref{eq:pr1cap2}, and \eqref{eq:pr1cap3}, the result follows. 

\end{proof}

\subsection{Proof of Propositions in Section ~\ref{subsec:checks}}\label{app:proofsecII2}

\begin{proof}[Proof of Proposition~\ref{prop:symmalpha}]
By definition, we have that
\begin{equation*}
\begin{split}
\bar{\alpha}\Ldens{a}_{\rm p}^+ (y) \Delta y &\approx \int\limits_{t \in L^{-1}_{\rm p}([y, y+\Delta y])} \bar{\alpha} W(t\mid 1) \,\,dt = \int\limits_{t \in L^{-1}_{\rm p}([y, y+\Delta y])} \alpha e^{L_{\rm p}(t)}W(t\mid -1) \,\,dt \\
& \approx e^y\int\limits_{t \in L^{-1}_{\rm p}([y, y+\Delta y])} \alpha W(t\mid -1) \,\,dt = e^y \alpha \Ldens{a}_{\rm p}^- (y) \Delta y,
\end{split}
\end{equation*}
where $L^{-1}_{\rm p}$ is the inverse of $L_{\rm p}$ defined in~\eqref{eq:loglikepost}. By taking $\Delta y\to 0$, we obtain that 
\begin{equation}\label{eq:reexpalpha}
\bar{\alpha}\Ldens{a}_{\rm p}^+(y)=e^y \alpha \Ldens{a}_{\rm p}^- (y).
\end{equation} 
With the change of variable $y\to -y$, we also obtain that
\begin{equation}\label{eq:reexp2alpha}
\alpha\Ldens{a}_{\rm p}^- (-y)=\bar{\alpha}e^y\Ldens{a}_{\rm p}^+ (-y).
\end{equation} 
As a result, condition \eqref{eq:symm} is fulfilled for $\Ldens{a}_{\rm p}^{\rm s}(y)$ and the statement follows. 
\end{proof}

\begin{proof}[Proof of Proposition~\ref{prop:capalpha}]
Since the log-likelihood ratio constitutes a sufficient statistic, two B-DMCs with non-uniform input distributions are equivalent if they have the same densities of the log-posterior ratio given that $X=\pm 1$ is transmitted. As a representative for the equivalence class, we can take 
\begin{equation}\label{eq:defcapalpha}
\begin{split}
W(y\mid 1) &= \Ldens{a}_{\rm p}^+(y),\\ 
W(y\mid -1) &= \Ldens{a}_{\rm p}^-(y).\\ 
\end{split}
\end{equation} 
By definition of log-posterior ratio and by using \eqref{eq:reexpalpha}, we have
\begin{equation*}\label{eq:loglikedensalpha}
L_{\rm p}(y) = \ln \frac{p_{X\mid Y}(1\mid y)}{p_{X\mid Y}(-1\mid y)} = \ln \frac{\bar{\alpha}\Ldens{a}_{\rm p}^+(y)}{\alpha\Ldens{a}_{\rm p}^-(y)} = y. 
\end{equation*}
Therefore,
\begin{equation*}
\lim_{\Delta y\to 0} \frac{{\mathbb P}(L_{\rm p}(Y)\in [y, y+\Delta y]\mid X = \pm 1)}{\Delta y} = \Ldens{a}_{\rm p}^{\pm}(y),
\end{equation*}
which means that \eqref{eq:defcapalpha} is a valid choice. 

Let $X\in \{-1, 1\}$ be s.t. ${\mathbb P}(X=-1) = \alpha$. Then, after some calculations we have that
\begin{equation}\label{eq:pr1capalpha}
\begin{split}
H(X \mid  Y) = - \int \bar{\alpha} W(y\mid 1)& \log \frac{\bar{\alpha}W(y\mid 1)}{\bar{\alpha}W(y\mid 1)+\alpha W(y\mid -1)}dy 
- \int \alpha W(y\mid -1) \log \frac{\alpha W(y\mid -1)}{\bar{\alpha} W(y\mid 1)+\alpha W(y\mid -1)}dy.
\end{split}
\end{equation}
By applying \eqref{eq:defcapalpha} and \eqref{eq:reexpalpha}, the first integral simplifies to 
\begin{equation}\label{eq:pr1cap2alpha}
\int \bar{\alpha}\Ldens{a}_{\rm p}^+(y) \left(1-\log (1+e^{-y})\right)\,\,dy.
\end{equation}
By applying \eqref{eq:defcapalpha}, doing the change of variables $y\to -y$ and using \eqref{eq:reexp2alpha}, the second integral simplifies to 
\begin{equation}\label{eq:pr1cap3alpha}
\int \alpha\Ldens{a}_{\rm p}^-(-y) \left(1-\log (1+e^{-y})\right)\,\,dy.
\end{equation}
By combining \eqref{eq:pr1capalpha}, \eqref{eq:pr1cap2alpha}, and \eqref{eq:pr1cap3alpha}, the result follows. 

\end{proof}

\subsection{Proof of Proposition~\ref{prop:boundmi}}\label{app:boundmi}

Before starting with the proof of the proposition, let us state the
following useful result \cite{ZZ07} that is a refinement of \cite[Lemma 2.7]{CS11}.
\begin{lemma}\label{lm:CK}
Consider two distributions $p$ and $p^*$ over the alphabet $\mathcal
X$ s.t. their total variation distance is equal to $\delta$, i.e., $\frac{1}{2}\sum_{x\in {\mathcal X}}|p^*(x)-p(x)| = \delta$. Take $X\sim p$ and $X^*\sim p^*$. Then,
\begin{equation}
\begin{split}
|H(X^*)-H(X)| &\le \delta \log (|{\mathcal X}|-1) +h_2(\delta)< \delta \log |{\mathcal X}| +h_2(\delta). 
\end{split}
\end{equation}
\end{lemma}

\begin{proof}[Proof of Proposition~\ref{prop:boundmi}]
Let $X\sim p$, $X^*\sim p^*$ and denote by $Y\sim p_Y$ and $Y^* \sim p_Y^*$ the outputs of the channel when the input is $X$ and $X^*$, respectively. Denote by $W(y\mid x)$ the probability distribution associated with the channel $W$. In order to prove \eqref{eq:boundY}, we write 
\begin{equation}\label{eq:basic1}
|I(p^*)-I(p)| \le |H(Y^*)-H(Y)|+|H(Y^*\mid X^*)-H(Y \mid X)|,
\end{equation}
and we bound both terms as functions of $\delta$ and $|\mathcal
Y|$. For the first term, observe that
\begin{equation}\label{eq:boundpy}
\begin{split}
\frac{1}{2}\sum_{y\in {\mathcal Y}}&|p_{Y}^*(y)-p_Y(y)| \le \frac{1}{2} \sum_{y\in {\mathcal Y}}\sum_{x\in {\mathcal X}} W(y\mid x) |p^*(x)-p(x)| <\delta,
\end{split}
\end{equation}
where it is used the fact that $\sum_{y\in {\mathcal Y}}W(y\mid x)=1$
for any $x\in \mathcal X$. Then, by using Lemma \ref{lm:CK} and the fact that $h_2(\delta)$ is increasing for any $\delta \in (0, 1/2)$, we obtain that
\begin{equation}\label{eq:inter11}
|H(Y^*)-H(Y)| < \delta \log |{\mathcal Y}| +h_2(\delta).
\end{equation}
For the second term, observe that the conditional distribution of
$Y^*$ given $X^*=x$ and the conditional distribution of $Y$ given
$X=x$ are both equal to $W(y\mid x)$. Therefore,
\begin{equation*}
H(Y \mid X=x) = H(Y^* \mid X^*=x) \le \log |{\mathcal Y}|.
\end{equation*}
Consequently, 
\begin{equation}\label{eq:inter12}
\begin{split}
|H(Y^*\mid X^*)-H(Y\mid X)| &\le \sum_{x\in {\mathcal X}}|p^*(x)-p(x)| H(Y \mid X=x) < 2\delta \log|{\mathcal Y}|.
\end{split}
\end{equation}
By combining \eqref{eq:basic1} with \eqref{eq:inter11} and
\eqref{eq:inter12}, we obtain the desired result.

In order to prove \eqref{eq:boundX}, we write
\begin{equation}\label{eq:basic2}
|I(p^*)-I(p)| \le |H(X^*)-H(X)|+|H(X^*\mid Y^*)-H(X \mid Y)|,
\end{equation}
and we bound both terms with functions of $\delta$ and $|\mathcal
X|$. The first term is easily bounded by using Lemma \ref{lm:CK} and the fact that $h_2(\delta)$ is increasing for any $\delta \in (0, 1/2)$,
\begin{equation}\label{eq:inter21}
|H(X^*)-H(X)| < \delta \log |{\mathcal X}| +h_2(\delta).
\end{equation} 
For the second term, consider the conditional distribution of $X^*$
given $Y^*=y$, i.e., $p_{X\mid Y}^*(x\mid y) = p^*(x)W(y\mid
x)/p_Y^*(y)$, and the conditional distribution of $X$ given $Y=y$,
i.e., $p_{X\mid Y}(x\mid y) = p(x)W(y\mid x)/p_{Y}(y)$.
Then,
\begin{equation}\label{eq:inter22}
\begin{split}
|H&(X^*\mid Y^*)-H(X\mid Y)| =|\sum_{y\in {\mathcal Y}} p_{Y}^*(y) H(X^*|Y^* = y) - p_{Y}(y) H(X|Y = y)| \\
&\le|\sum_{y\in {\mathcal Y}} p_{Y}^*(y) H(X^*|Y^* = y) - p_{Y}(y) H(X^*|Y^* = y)| +|\sum_{y\in {\mathcal Y}} p_{Y}(y) H(X^*|Y^* = y) - p_{Y}(y) H(X|Y = y)|.
\end{split}  
\end{equation}
In order to bound the first term of \eqref{eq:inter22}, observe
that $H(X^*|Y^* = y)\le \log|\mathcal X|$ for any $y\in \mathcal
Y$. Therefore, by using \eqref{eq:boundpy}, we obtain 
\begin{equation}\label{eq:inter221}
|\sum_{y\in {\mathcal Y}} p_{Y}^*(y) H(X^*|Y^* = y) - p_{Y}(y) H(X^*|Y^* = y)|
< 2\delta \log |\mathcal X|.
\end{equation}
For the second term of \eqref{eq:inter22}, let us denote by $d(y)$
the total variation distance between $p_{X\mid Y}^*(x\mid y)$ and $p_{X\mid
Y}(x\mid y)$, namely,
\begin{equation*}
d(y) = \frac{1}{2}\sum_{x\in {\mathcal X}} |p_{X\mid Y}^*(x\mid y)-p_{X\mid Y}(x\mid y)|.
\end{equation*} 
Then, by Lemma \ref{lm:CK},
\begin{equation*}
|H(X^*|Y^* = y)-H(X|Y = y)| < d(y) \log |\mathcal X| +h_2(d(y)),
\end{equation*}
which implies that
\begin{equation}
\begin{split}\label{eq:inter222}
|\sum_{y\in {\mathcal Y}} p_{Y}(y) H(X^*|Y^* = y) - p_{Y}(y) H(X|Y = y)|< \log|\mathcal X| \sum_{y\in {\mathcal Y}} p_Y(y)d(y) +\sum_{y\in {\mathcal Y}} p_Y(y)h_2(d(y)).
\end{split}
\end{equation}
Now, let us focus on the quantity $\sum_{y\in {\mathcal Y}} p_Y(y)d(y)$:
\begin{equation*}
\begin{split}
\sum_{y\in {\mathcal Y}} p_Y(y)d(y) &= \sum_{y\in {\mathcal Y}} \sum_{x\in {\mathcal X}} |p_Y(y)p^*_{X\mid Y}(x|y)-p_Y(y)p_{X\mid Y}(x|y)| \\
&\le \sum_{y\in {\mathcal Y}} \sum_{x\in {\mathcal X}} |p_Y(y)p^*_{X\mid Y}(x|y)-p^*_Y(y)p^*_{X\mid Y}(x|y)|+\sum_{y\in {\mathcal Y}} \sum_{x\in {\mathcal X}} |p^*_Y(y)p^*_{X\mid Y}(x|y)-p_Y(y)p_{X\mid Y}(x|y)|\\
&= \sum_{y\in {\mathcal Y}} |p^*_Y(y)-p_Y(y)|\sum_{x\in {\mathcal X}} p^*_{X\mid Y}(x|y)+\sum_{x\in {\mathcal X}} |p^*(x)-p(x)|\sum_{y\in {\mathcal Y}} W(y\mid x) < 4\delta.
\end{split}
\end{equation*}
Observe that $h_2(t)$ is concave for any
$t \in (0, 1)$ and increasing for $t \le 1/2$. Then, as $\delta < 1/8$,
\begin{equation*}
\begin{split}
\sum_{y\in {\mathcal Y}} p(y) h_2(d(y))\le h_2(\sum_{y\in {\mathcal Y}} p(y)d(y)) < h_2(4\delta).
\end{split}
\end{equation*}
By combining \eqref{eq:basic2} with \eqref{eq:inter21}, \eqref{eq:inter22}, \eqref{eq:inter221}, and \eqref{eq:inter222}, the result follows. 
\end{proof}

\bibliographystyle{IEEEtran}
\bibliography{lth,lthpub}

\newcommand{\SortNoop}[1]{}
\begin{thebibliography}{10}
\providecommand{\url}[1]{#1}
\csname url@samestyle\endcsname
\providecommand{\newblock}{\relax}
\providecommand{\bibinfo}[2]{#2}
\providecommand{\BIBentrySTDinterwordspacing}{\spaceskip=0pt\relax}
\providecommand{\BIBentryALTinterwordstretchfactor}{4}
\providecommand{\BIBentryALTinterwordspacing}{\spaceskip=\fontdimen2\font plus
\BIBentryALTinterwordstretchfactor\fontdimen3\font minus
  \fontdimen4\font\relax}
\providecommand{\BIBforeignlanguage}[2]{{%
\expandafter\ifx\csname l@#1\endcsname\relax
\typeout{** WARNING: IEEEtran.bst: No hyphenation pattern has been}%
\typeout{** loaded for the language `#1'. Using the pattern for}%
\typeout{** the default language instead.}%
\else
\language=\csname l@#1\endcsname
\fi
#2}}
\providecommand{\BIBdecl}{\relax}
\BIBdecl

\bibitem{Ari09}
E.~{Ar\i kan}, ``Channel polarization: A method for constructing
  capacity-achieving codes for symmetric binary-input memoryless channels,''
  \emph{IEEE Trans. on Inform. Theory}, vol.~55, no.~7, pp. 3051--3073, July
  2009.

\bibitem{KRU13}
S.~Kudekar, T.~Richardson, and R.~L. Urbanke, ``Spatially coupled ensembles
  universally achieve capacity under belief propagation,'' \emph{IEEE Trans.
  Inform. Theory}, vol.~59, no.~12, pp. 7761--7813, Dec. 2013.

\bibitem{mondelli-polarRM}
M.~Mondelli, S.~H. Hassani, and R.~L. Urbanke, ``From polar to {R}eed-{M}uller
  codes: {A} technique to improve the finite-length performance,'' \emph{IEEE
  Trans. Commun.}, vol.~62, no.~9, pp. 3084--3091, Sept. 2014.

\bibitem{TV13con}
I.~Tal and A.~Vardy, ``How to construct polar codes,'' \emph{IEEE Trans. on
  Inform. Theory}, vol.~59, no.~10, pp. 6562--6582, Oct. 2013.

\bibitem{RHTT}
R.~Pedarsani, S.~H. Hassani, I.~Tal, and E.~Telatar, ``On the construction of
  polar codes,'' in \emph{Proc. of the IEEE Int. Symposium on Inform. Theory},
  St. Petersberg, Russia, 2011, conference, pp. 11--15.

\bibitem{MHU17subl}
M.~Mondelli, S.~H. Hassani, and R.~L. Urbanke, ``Construction of polar codes
  with sublinear complexity,'' in \emph{Proc. of the IEEE Int. Symposium on
  Inform. Theory}, Aachen, Germany, June 2017, pp. 1853--1857.

\bibitem{ArT09}
E.~{Ar\i kan} and I.~E. {Telatar}, ``{On the rate of channel polarization},''
  in \emph{Proc. of the IEEE Int. Symposium on Inform. Theory}, Seoul, South
  Korea, July 2009, pp. 1493--1495.

\bibitem{HAU14}
S.~H. Hassani, K.~Alishahi, and R.~L. Urbanke, ``Finite-length scaling for
  polar codes,'' \emph{IEEE Trans. Inform. Theory}, vol.~60, no.~10, pp. 5875
  -- 5898, Oct. 2014.

\bibitem{GB14}
D.~Goldin and D.~Burshtein, ``Improved bounds on the finite length scaling of
  polar codes,'' \emph{IEEE Trans. Inform. Theory}, vol.~60, no.~11, pp.
  6966--6978, Nov. 2014.

\bibitem{MHU15}
M.~Mondelli, S.~H. Hassani, and R.~L. Urbanke, ``Unified scaling of polar
  codes: {E}rror exponent, scaling exponent, moderate deviations, and error
  floors,'' \emph{IEEE Trans. Inform. Theory}, vol.~62, no.~12, pp. 6698--6712,
  Dec. 2016.

\bibitem{OU14}
P.~M. Olmos and R.~L. Urbanke, ``A scaling law to predict the finite-length
  performance of spatially-coupled {LDPC} codes,'' \emph{IEEE Trans. Inform.
  Theory}, vol.~61, no.~6, pp. 3164--3184, June 2015.

\bibitem{STA09}
E.~{\c Sa\c so\u glu}, I.~E. Telatar, and E.~Ar{\i}kan, ``Polarization for
  arbitrary discrete memoryless channels,'' in \emph{Proc. of the IEEE Inform.
  Theory Workshop}, Taormina, Italy, Oct. 2009, pp. 144--148.

\bibitem{MT10}
R.~Mori and T.~Tanaka, ``Channel polarization on $q$-ary discrete memoryless
  channels by arbitrary kernel,'' in \emph{Proc. of the IEEE Int. Symposium on
  Inform. Theory}, Austin, USA, June 2010, pp. 894--898.

\bibitem{PB13}
W.~Park and A.~Barg, ``Polar codes for $q$-ary channels, $q= 2^r$,'' \emph{IEEE
  Trans. on Inform. Theory}, vol.~59, no.~2, pp. 955--969, Feb. 2013.

\bibitem{SaP13}
A.~G. Sahebi and S.~S. Pradhan, ``Multilevel channel polarization for arbitrary
  discrete memoryless channels,'' \emph{IEEE Trans. on Inform. Theory},
  vol.~59, no.~12, pp. 7839--7857, Dec. 2013.

\bibitem{NT16}
R.~Nasser and I.~E. Telatar, ``Polar codes for arbitrary {DMC}s and arbitrary
  {MAC}s,'' \emph{IEEE Trans. Inform. Theory}, vol.~62, no.~6, pp. 2917--2936,
  June 2016.

\bibitem{Nas17a}
R.~Nasser, ``An ergodic theory of binary operations -- {P}art {I}: {K}ey
  properties,'' \emph{IEEE Trans. Inform. Theory}, vol.~62, no.~12, pp.
  6931--6952, Dec. 2016.

\bibitem{Nas17b}
------, ``An ergodic theory of binary operations -- {Part II: A}pplications to
  polarization,'' \emph{IEEE Trans. Inform. Theory}, vol.~63, no.~2, pp.
  1063--1083, Feb. 2017.

\bibitem{UKS11}
H.~Uchikawa, K.~Kasai, and K.~Sakaniwa, ``Design and performance of
  rate-compatible non-binary {LDPC} convolutional codes,'' \emph{{IEICE Trans.
  on Fundamentals of Electronics, Communications and Computer Sciences}},
  vol.~94, no.~11, pp. 2135--2143, Nov. 2011.

\bibitem{PAC13}
A.~Piemontese, A.~G. Amat, and G.~Colavolpe, ``Nonbinary spatially-coupled
  {LDPC} codes on the binary erasure channel,'' in \emph{Proc. of the IEEE Int.
  Conf. Commun.}, Budapest, Hungary, June 2013, pp. 3270--3274.

\bibitem{AA16}
I.~Andriyanova and A.~G. Amat, ``Threshold saturation for nonbinary {SC-LDPC}
  codes on the binary erasure channel,'' \emph{IEEE Trans. on Inform. Theory},
  vol.~62, no.~5, pp. 2622--2638, May 2016.

\bibitem{WKMFC14}
L.~Wei, T.~Koike-Akino, D.~G.~M. Mitchell, T.~E. Fuja, and D.~J. {Costello Jr},
  ``Threshold analysis of non-binary spatially-coupled {LDPC} codes with
  windowed decoding,'' in \emph{Proc. of the IEEE Int. Symposium on Inform.
  Theory}, Honolulu, HI, USA, July 2014, pp. 881--885.

\bibitem{RiU08}
T.~Richardson and R.~L. Urbanke, \emph{Modern Coding Theory}.\hskip 1em plus
  0.5em minus 0.4em\relax Cambridge University Press, 2008.

\bibitem{MaR91}
E.~E. Majani and H.~Rumsey, Jr., ``Two results on binary-input discrete
  memoryless channels,'' in \emph{Proc. of the IEEE Int. Symposium on Inform.
  Theory}, Budapest, Hungary, June 1991, p. 104.

\bibitem{ShF04}
N.~Shulman and M.~Feder, ``The uniform distribution as a universal prior,''
  \emph{IEEE Trans. on Inform. Theory}, vol.~50, no.~6, pp. 1356--1362, June
  2004.

\bibitem{Lia04}
X.-B. Liang, ``On a conjecture of {M}ajani and {R}umsey,'' in \emph{Proc. of
  the IEEE Int. Symposium on Inform. Theory}, Chicago, USA, June 2004, p.~62.

\bibitem{Gal68}
R.~G. Gallager, \emph{Information Theory and Reliable Communication}.\hskip 1em
  plus 0.5em minus 0.4em\relax New York: Wiley, 1968.

\bibitem{McE01}
R.~J. McEliece, ``Are turbo-like codes effective on nonstandard channels?''
  \emph{IEEE Inform. Theory Soc. Newslett.}, vol.~51, no.~4, pp. 1--8, Dec.
  2001.

\bibitem{SRDR12}
D.~Sutter, J.~M. Renes, F.~Dupuis, and R.~Renner, ``Achieving the capacity of
  any {DMC} using only polar codes,'' in \emph{Proc. of the IEEE Inform. Theory
  Workshop}, Lausanne, Switzerland, Sept. 2012, pp. 114--118.

\bibitem{HY13}
J.~Honda and H.~Yamamoto, ``Polar coding without alphabet extension for
  asymmetric models,'' \emph{IEEE Trans. on Inform. Theory}, vol.~59, no.~12,
  pp. 7829--7838, Dec. 2013.

\bibitem{GAG13ar}
N.~Goela, E.~Abbe, and M.~Gastpar, ``Polar codes for broadcast channels,''
  \emph{IEEE Trans. on Inform. Theory}, vol.~61, no.~2, pp. 758--782, Feb.
  2015.

\bibitem{Ar10}
E.~{Ar\i kan}, ``{Source polarization},'' in \emph{Proc. of the IEEE Int.
  Symposium on Inform. Theory}, Austin, USA, June 2010, pp. 899--903.

\bibitem{CrKo10}
H.~S. Cronie and S.~B. Korada, ``Lossless source coding with polar codes,'' in
  \emph{Proc. of the IEEE Int. Symposium on Inform. Theory}, Austin, USA, June
  2010, pp. 904--908.

\bibitem{CSV04dimacs}
G.~Caire, S.~Shamai, and S.~Verd{\'u}, ``Noiseless data compression with
  low-density parity-check codes,'' \emph{DIMACS Series in Discrete Mathematics
  and Theoretical Computer Science}, vol.~66, pp. 263--284, 2004.

\bibitem{MuMi08}
S.~Miyake and J.~Muramatsu, ``A construction of lossy source code using {LDPC}
  matrices,'' \emph{{IEICE Trans. on Fundamentals of Electronics,
  Communications and Computer Sciences}}, vol. E91-A, no.~6, pp. 1488--1501,
  June 2008.

\bibitem{MuMi09}
------, ``A construction of channel code, joint source-channel code, and
  universal code for arbitrary stationary memoryless channels using sparse
  matrices,'' \emph{{IEICE Trans. on Fundamentals of Electronics,
  Communications and Computer Sciences}}, vol. E92-A, no.~9, pp. 2333--2344,
  Sept. 2009.

\bibitem{MuMi10}
J.~Muramatsu and S.~Miyake, ``Hash property and coding theorems for sparse
  matrices and maximum-likelihood coding,'' \emph{IEEE Trans. on Inform.
  Theory}, vol.~56, no.~5, pp. 2143--2167, May 2010.

\bibitem{AMV15}
V.~Aref, N.~Macris, and M.~Vuffray, ``Approaching the rate-distortion limit
  with spatial coupling, belief propagation, and decimation,'' \emph{IEEE
  Trans. Inform. Theory}, vol.~61, no.~7, pp. 3954 -- 3979, July 2015.

\bibitem{KVNP14}
S.~Kumar, A.~Vem, K.~Narayanan, and H.~D. Pfister, ``Spatially-coupled codes
  for side-information problems,'' in \emph{Proc. of the IEEE Int. Symposium on
  Inform. Theory}, Honolulu, HI, USA, July 2014, pp. 516 -- 520.

\bibitem{BoM11}
G.~B\"ocherer and R.~Mathar, ``Operating {LDPC} codes with zero shaping gap,''
  in \emph{Proc. of the IEEE Inform. Theory Workshop}, Paraty, Oct. 2011, pp.
  330--334.

\bibitem{HRunipol}
S.~H. Hassani and R.~L. Urbanke, ``Universal polar codes,'' in \emph{IEEE
  International Symposium on Information Theory (ISIT)}, 2014, pp. 1451--1455.

\bibitem{SaV13}
E.~{\c Sa\c so\u glu} and A.~Vardy, ``A new polar coding scheme for strong
  security on wiretap channels,'' in \emph{Proc. of the IEEE Int. Symposium on
  Inform. Theory}, Istanbul, Turkey, July 2013, pp. 1117--1121.

\bibitem{MHSU14}
M.~Mondelli, S.~H. Hassani, I.~Sason, and R.~L. Urbanke, ``Achieving {M}arton's
  region for broadcast channels using polar codes,'' \emph{IEEE Trans. Inform.
  Theory}, vol.~61, no.~2, pp. 783--800, Feb. 2015.

\bibitem{Huo03}
J.~Hou, P.~H. Siegel, L.~B. Milstein, and H.~D. Pfister, ``Capacity-approaching
  bandwidth-efficient coded modulation schemes based on low-density
  parity-check codes,'' \emph{IEEE Trans. on Inform. Theory}, vol.~49, no.~9,
  pp. 2141--2155, Sept. 2003.

\bibitem{WKP03}
C.-C. Wang, S.~R. Kulkarni, and H.~V. Poor, ``Density evolution for asymmetric
  memoryless channels,'' \emph{IEEE Trans. on Inform. Theory}, vol.~51, no.~12,
  pp. 4216--4236, Dec. 2005.

\bibitem{W15al}
L.~Wang and Y.-H. Kim, ``Linear code duality between channel coding and
  {S}lepian-{W}olf coding,'' in \emph{Proc. of the Allerton Conf. on Commun.,
  Control, and Computing}, Monticello, IL, USA, Oct. 2015, pp. 147--152.

\bibitem{KaSa11}
K.~Kasai and K.~Sakaniwa, ``{Spatially-coupled MacKay-Neal codes and
  Hsu-Anastasopoulos codes},'' \emph{{IEICE Trans. on Fundamentals of
  Electronics, Communications and Computer Sciences}}, vol.~94, no.~11, pp.
  2161--2168, Nov. 2011.

\bibitem{MKLC12}
D.~G.~M. Mitchell, K.~Kasai, M.~Lentmaier, and D.~J. Costello, Jr.,
  ``{Asymptotic analysis of spatially coupled MacKay-Neal and
  Hsu-Anastasopoulos LDPC codes},'' in \emph{Proc. of the IEEE Int. Symposium
  on Inform. Theory and its Applications}, Honolulu, HI, USA, Oct. 2012, pp.
  337 -- 341.

\bibitem{SS03allerton}
J.~B. Soriaga and P.~H. Siegel, ``On distribution shaping codes for
  partial-response channels,'' in \emph{Proc. of the Allerton Conf. on Commun.,
  Control, and Computing}, Monticello, IL, USA, 2003.

\bibitem{AbT12}
E.~Abbe and I.~E. Telatar, ``Polar codes for the $m$-user multiple access
  channel,'' \emph{IEEE Trans. Inform. Theory}, vol.~58, no.~8, pp. 5437--5448,
  Aug. 2012.

\bibitem{Hon13}
J.~Honda, ``Efficient polar and {LDPC} coding for asymmetric channels and
  sources,'' Ph.D. dissertation, The University of Tokyo, Tokyo, Japan, Mar.
  2013.

\bibitem{KoU09}
S.~B. Korada and R.~L. Urbanke, ``Polar codes are optimal for lossy source
  coding,'' \emph{IEEE Trans. Inform. Theory}, vol.~56, no.~4, pp. 1751--1768,
  Apr. 2010.

\bibitem{CB15}
R.~A. Chou and M.~R. Bloch, ``Using deterministic decisions for low-entropy
  bits in the encoding and decoding of polar codes,'' in \emph{Proc. of the
  Allerton Conf. on Commun., Control, and Computing}, Monticello, IL, USA, Oct.
  2015, pp. 1380--1385.

\bibitem{CM05}
S.~Ciliberti and M.~M{\'e}zard, ``The theoretical capacity of the parity source
  coder,'' \emph{Journal of Statistical Mechanics: Theory and Experiment},
  no.~10, Oct. 2005.

\bibitem{Mu14}
J.~Muramatsu, ``Channel coding and lossy source coding using a generator of
  constrained random numbers,'' \emph{IEEE Trans. on Inform. Theory}, vol.~60,
  no.~5, pp. 2667--2686, May 2014.

\bibitem{CoT06}
T.~M. Cover and J.~A. Thomas, \emph{Elements of Information Theory}.\hskip 1em
  plus 0.5em minus 0.4em\relax New York: Wiley, 2006.

\bibitem{Gu88}
C.~G. G{\"u}ther, ``A universal algorithm for homophonic coding,'' in
  \emph{Advances in Cryptology -- EUROCRYPT '88}, ser. Lecture Notes in
  Computer Science Series.\hskip 1em plus 0.5em minus 0.4em\relax
  Springer-Verlag, 1988, pp. 405--414.

\bibitem{JKM89}
H.~N. Jendal, Y.~J.~B. Kuhn, and J.~L. Massey, ``An information-theoretic
  treatment of homophonic substitution,'' in \emph{Advances in Cryptology --
  EUROCRYPT '89}, ser. Lecture Notes in Computer Science Series.\hskip 1em plus
  0.5em minus 0.4em\relax Springer-Verlag, 1989, pp. 382--394.

\bibitem{HoHa01}
M.~Hoshi and T.~S. Han, ``Interval algorithm for homophonic coding,''
  \emph{IEEE Trans. on Inform. Theory}, vol.~47, no.~3, pp. 1021--1031, Mar.
  2001.

\bibitem{HoYa15}
R.~Wang, J.~Honda, H.~Yamamoto, R.~Liu, and Y.~Hou, ``Construction of polar
  codes for channels with memory,'' in \emph{Proc. of the IEEE Inform. Theory
  Workshop}, Jeju, South Korea, Oct. 2015, pp. 187--191.

\bibitem{TVa15}
I.~Tal and A.~Vardy, ``{List decoding of polar codes},'' \emph{IEEE Trans.
  Inform. Theory}, vol.~61, no.~5, pp. 2213--2226, May 2015.

\bibitem{hassani2012polar}
S.~H. Hassani and R.~L. Urbanke, ``Polar codes: Robustness of the successive
  cancellation decoder with respect to quantization,'' in \emph{Proc. of the
  IEEE Int. Symposium on Inform. Theory}, Cambridge, MA, USA, July 2012, pp.
  1962--1966.

\bibitem{fong2016scaling}
S.~L. Fong and V.~Y.~F. Tan, ``On the scaling exponent of polar codes for
  binary-input energy-harvesting channels,'' \emph{IEEE J. Select. Areas
  Commun.}, vol.~34, no.~12, pp. 3540--3551, Dec. 2016.

\bibitem{SaWa16}
E.~\c{S}a\c{s}o\u{g}lu and L.~Wang, ``Universal polarization,'' \emph{IEEE
  Trans. Inform. Theory}, vol.~62, no.~6, pp. 2937--2946, June 2016.

\bibitem{Mahdi10thesis}
M.~Cheraghchi, ``Applications of derandomization in coding theory,'' Ph.D.
  dissertation, EPFL, Lausanne, Switzerland, July 2010.

\bibitem{Ab15}
E.~Abbe, ``Randomness and dependencies extraction via polarization, with
  applications to slepian-wolf coding and secrecy,'' \emph{IEEE Trans. on
  Inform. Theory}, vol.~61, no.~5, pp. 2388--2398, May 2015.

\bibitem{ZZ07}
Z.~Zhang, ``Estimating mutual information via {K}olmogorov distance,''
  \emph{IEEE Trans. on Inform. Theory}, vol.~53, no.~9, pp. 3280--3282, Sept.
  2007.

\bibitem{CS11}
I.~Csisz{\'a}r and J.~K{\"o}rner, \emph{Information theory: Coding Theorems for
  Discrete Memoryless Systems}.\hskip 1em plus 0.5em minus 0.4em\relax
  Cambridge University Press, 2011.

\end{thebibliography}


\end{document}